\documentclass[10pt,journal,compsoc]{IEEEtran}
\IEEEoverridecommandlockouts

\usepackage{amsmath,amssymb,amsfonts,amsthm}
\usepackage[linesnumbered,ruled,vlined]{algorithm2e}
\SetArgSty{textup}
\usepackage{array}
\usepackage{bbding}
\usepackage{bm}
\usepackage{booktabs}
\usepackage{enumitem}
\usepackage{extarrows}
\usepackage{graphicx}
\usepackage{hyperref}
\usepackage{listing}
\usepackage{makecell}
\usepackage{multirow}
\usepackage{pifont}
\usepackage{subcaption}
\usepackage{threeparttable}
\usepackage[most]{tcolorbox}

\newtcolorbox{mybox}[2][]{%
	attach boxed title to top left
	= {xshift=10pt, yshift=-10pt},
	colframe     = black,
	colback      = white,
	coltitle     = black,
	colbacktitle = white,
	title        = #2,#1,
	enhanced,
	left=2pt,
	top=10pt,
}
\usepackage{textcomp}
\newtheorem{theorem}{Theorem}
\newtheorem{lemma}{Lemma}
\usepackage{tikz}
\usetikzlibrary{positioning}
\usepackage[normalem]{ulem}
\usepackage{xcolor}

\newcommand{\na}{N/A}
\newcommand{\xmark}{\ding{55}}%

\newcommand{\role}[1]{$\triangleright$ {\color{blue}#1}}

\newtheorem{problem}{Question}

\ifCLASSOPTIONcompsoc
\usepackage{cite}
\else
\usepackage{cite}
\fi

\ifCLASSINFOpdf

\else

\fi

\makeatletter
\def\ps@IEEEtitlepagestyle{%
  \def\@oddfoot{\mycopyrightnotice}%
  \def\@evenfoot{}%
}
\def\mycopyrightnotice{%
  {\footnotesize 
  \begin{minipage}{\textwidth}
  \centering
  \textbf{This work has been submitted to the IEEE for possible publication. Copyright may be transferred without notice, after which this version may no longer be accessible.}
  \end{minipage}
    \hfill}% <--- Change here
  \gdef\mycopyrightnotice{}% just in case
}

\begin{document}
	
	\title{BFT-DSN: A Byzantine Fault Tolerant Decentralized Storage Network}

	\author{Hechuan Guo,
		Minghui Xu,
		Jiahao Zhang,
            Chunchi Liu,
            Rajiv Ranjan,
		Dongxiao Yu,
            Xiuzhen Cheng
    
		\IEEEcompsocitemizethanks{\IEEEcompsocthanksitem H. Guo, M. Xu, J. Zhang, D. Yu and X. Cheng are with the School of Computer and Science and Technology, Shandong University. Email: \{ghc, zjh\}@mail.sdu.edu.cn, \{mhxu, dxyu, xzcheng\}@sdu.edu.cn\\

        \IEEEcompsocthanksitem C. Liu is with the Huawei Technologies Co., Ltd. E-mail: liuchunchi@huawei.com\\

        \IEEEcompsocthanksitem R. Ranjan is with the School of Computing, Newcastle University, Newcastle, United Kingdom. E-mail: raj.ranjan@newcastle.ac.uk\\

        \IEEEcompsocthanksitem Corresponding author: Minghui Xu.
            }% <-this % stops an unwanted space
	}
	
	%{Shell \MakeLowercase{\textit{et al.}}: Bare Demo of IEEEtran.cls for Computer Society Journals}
	
	\IEEEtitleabstractindextext{%
		\begin{abstract}
            With the rapid development of blockchain and its applications, the amount of data stored on decentralized storage networks (DSNs) has grown exponentially. DSNs bring together affordable storage resources from around the world to provide robust, decentralized storage services for tens of thousands of decentralized applications (dApps). However, existing DSNs do not offer verifiability when implementing erasure coding for redundant storage, making them vulnerable to Byzantine encoders. Additionally, there is a lack of Byzantine fault-tolerant consensus for optimal resilience in DSNs. This paper introduces BFT-DSN, a Byzantine fault-tolerant decentralized storage network designed to address these challenges. BFT-DSN combines storage-weighted BFT consensus with erasure coding and incorporates homomorphic fingerprints and weighted threshold signatures for decentralized verification. The implementation of BFT-DSN demonstrates its comparable performance in terms of storage cost and latency as well as superior performance in Byzantine resilience when compared to existing industrial decentralized storage networks. % This paper presents BFT-DSN as an efficient solution for decentralized applications, offering improved fault tolerance, and contributing to the advancement of decentralized storage technologies.	
            \end{abstract}
		
		}
	
	\maketitle
	
	\IEEEdisplaynontitleabstractindextext
	
	\IEEEpeerreviewmaketitle

\section{Introduction}\label{sec:introduction}

Decentralized storage networks (DSNs) have been developed with the major goal of establishing robust management of storage security within a zero trust environment. These networks make use of underutilized storage resources, strategically distributed across various devices and locations, to optimize storage efficiency. Within a DSN, redundant storage mechanisms are implemented to enhance fault tolerance and maintain continuous data availability. Erasure coding (EC) is such a widely adopted and highly efficient technique. It partitions data into smaller pieces, generates redundant fragments, then disperses them across multiple storage nodes. Its design can enable data recovery in the event of failures or data losses. Notably, DSNs deviate from conventional storage networks by operating on top of a blockchain system, which serves as an incentivizing layer. Miners are incentivized to provide reliable storage services to clients, thereby fostering an open and manageable storage marketplace. Additionally, the blockchain can function as a state machine replication protocol, ensuring the integrity of file storage against Byzantine behaviors. However, current EC-based DSNs overlook the following two critical issues that significantly impact their security and performance.

\begin{problem}
How to guarantee security when applying erasure coding in a DSN against Byzantine attackers?
\end{problem}
In comparison to traditional methods of data replication, EC offers improved storage utilization while maintaining fault tolerance. Prominent examples of DSNs, such as Sia \cite{sia} and Storj \cite{storj}, utilize EC to reduce redundancy. However, these systems cannot guarantee the integrity of the EC encoding process confronting Byzantine encoders. As shown in Fig.~\ref{fig:byzantine_encoder}, if an encoder maliciously alters the data, inconsistencies between the decoded and encoded files \cite{ecp} may appear. This attack may remain undetected before the decoding process ends, potentially disrupting the storage service. BFT-Store \cite{bftstore} introduces Byzantine fault tolerance to address such attacks. However, it requires that each data chunk be encoded independently and locally by every encoder at the outset. This requirement is resource-intensive in terms of storage and bandwidth, and may even be impractical in DSNs where data chunks, such as videos, music, and photo galleries, are often too large to be transferred among encoders. In addition, private data cannot be shared at will to meet this requirement. Therefore, it is necessary to reconsider DSN-oriented EC methods that can defend against Byzantine attacks while maintaining reasonable encoding and decoding efficiency.

\begin{figure}[!htbp]
    \centering
    \includegraphics[width=0.48\textwidth]{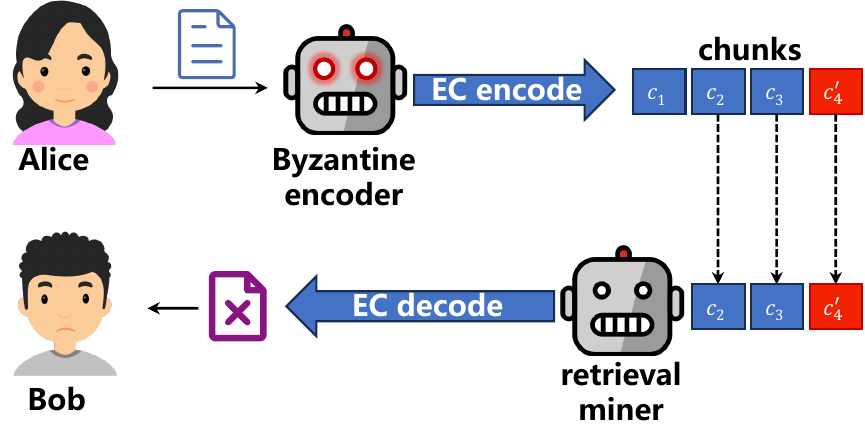}
    \caption{An attack on EC-based storage by a Byzantine encoder}
    \label{fig:byzantine_encoder}
\end{figure}

\begin{problem}
How to practically optimize the resiliency of BFT consensus in a DSN?
\end{problem}
To the best of our knowledge, the Expected Consensus mechanism employed by Filecoin is the sole Byzantine Fault Tolerant (BFT) consensus protocol specifically tailored for DSNs. However, its capability to withstand Byzantine adversaries is limited to cases where the adversary controls no more than 20\% of the total storage capacity pledged to the network \cite{ECConsensusAnal}. Nevertheless, the optimal resilience level of BFT consensus is 33\% in theory. While there exist multiple BFT consensus algorithms that can handle Byzantine adversaries controlling no more than one-third of the network nodes \cite{pbft, tendermint, honeybadger, hotstuff}, directly applying them to DSNs poses challenges, primarily due to their vulnerability to Sybil attacks in an open network environment. Specifically, considering all miners as equal entities enables the possibility of Sybil attacks, wherein an attacker can register multiple identities without possessing significant physical storage resources to surpass the BFT consensus threshold. Consequently, it is imperative for a BFT consensus algorithm to be adaptable to the practical scenario of a DSN.

To tackle these challenges, we introduce BFT-DSN, a DSN that can ensure Byzantine fault tolerance for both EC and consensus. In BFT-DSN, we devise a publicly verifiable EC scheme that tolerates up to $\lfloor \frac{n-1}{3}\rfloor$ Byzantine faults, where $n$ represents the total number of sectors pledged in a DSN. Moreover, BFT-DSN incorporates a novel Storage-Weighted BFT (SW-BFT) consensus algorithm that assesses Byzantine adversaries based on their storage resource ratio rather than the number of consensus nodes they possess, to address the Sybil attack issue. The contributions are summarized as follows:
\begin{enumerate}
    \item To the best of our knowledge, BFT-DSN is the first DSN that achieves both BFT erasure coding and BFT consensus with optimal resiliency. Specifically, we integrate erasure coding, BFT consensus, and threshold signature to unify their security upper bound to $\lfloor \frac{n-1}{3}\rfloor$. 
    \item We propose a secure erasure coding scheme against Byzantine encoders, leveraging homomorphic fingerprints and threshold signatures to achieve public verifiability.
    \item BFT-DSN supports optimal resilience when $f = \lfloor \frac{n-1}{3}\rfloor$. To achieve this goal, we introduce a storage-weighted BFT consensus algorithm and a proof of storage mechanism that audits the storage power of each miner.
    \item Last but not the least, we implement BFT-DSN and evaluate its performance through extensive experiments. The results demonstrate the superior performance of our BFT-DSN in storage cost and operation latency compared to the state-of-the-art DSNs, while providing a robust security guarantee.
\end{enumerate}

The rest of this paper is organized as follows. Section~\ref{sec:related} summarizes the most related work, presents preliminary knowledge, and explains the problem settings. Section~\ref{sec:design} details our BFT-DSN design and demonstrates how it works. Key properties and performance evaluation results of BFT-DSN are respectively reported in Section~\ref{sec:analysis} and Section~\ref{sec:exp}. Finally, we summarize this paper in Section~\ref{sec:conclusion}.

\section{Related Work, Preliminaries, and Models}
\label{sec:related}
\subsection{Related Work}

Based on the approaches of realizing storage redundancy, we categorize existing DSNs into two types: replication-based and erasure coding-based. We also examine the research effort in applying erasure coding to distributed networked storage. Finally, we provide a brief analysis of the problems faced by the current solutions.

%Decentralized storage networks provide a secure and reliable solution for storing and sharing data. They utilize erasure codes to prevent data loss and Byzantine fault tolerance consensus to handle faulty nodes. These techniques are essential for the effective operation and maintenance of decentralized storage networks while ensuring data integrity. This section offers an overview of the existing literature on Byzantine fault tolerance and erasure coding.

\textbf{Replication-based DSN.}
Swarm \cite{swarm} was developed based on Ethereum \cite{ethereum}. It achieves data redundancy via the so-called pull-sync protocol, which duplicates and syncs file chunks to multiple miners. Swarm utilizes Proof-of-Stake consensus\footnote{Since September 15th, 2022, Ethereum has switched its consensus protocol from Proof-of-Work to Proof-of-Stake.}, but the Byzantine fault tolerance of this consensus protocol has not been proven. Filecoin \cite{filecoin} generates file replicas to provide redundant storage. It utilizes Proof-of-Replication (PoRep) and Proof-of-Space-Time (PoST) to verify the retrievability of the replicas. Filecoin's consensus protocol is called Expected Consensus, where the probability of a miner proposing a valid block is proportional to the amount of storage resources that miner contributes to the Filecoin network. However, research has shown that Filecoin's Expected Consensus can only tolerate up to 20\% of the storage resources being controlled by Byzantine adversaries \cite{ECConsensusAnal}. FileDAG \cite{filedag} provides multi-version file storage in a cost-effective and time-efficient manner. It stores multi-version files in the form of increments and provides redundant storage by duplicating these increments. To achieve consensus, FileDAG utilizes a two-layer DAG-based blockchain ledger inspired by DAG-Rider \cite{dagrider}. DAG-Rider is an asynchronous Byzantine Atomic Broadcast protocol that can tolerate up to 33\% of the nodes being Byzantine. Note that FileDAG's implementation of DAG-Rider does not take into account the varying importance of miners who contribute different amounts of storage resources. This oversight allows Byzantine adversaries to register multiple identities without sufficient physical storage resources, potentially bypassing the BFT consensus threshold.

\textbf{Erasure Coding-based DSN.}
Sia \cite{sia} is a Decentralized Storage Network that makes use of EC to distribute and store data across a network of storage miners. It employs Reed-Solomon erasure coding to generate redundant chunks of data that are then distributed to different hosts for fault tolerance. In terms of consensus, Sia adopts Bitcoin's Proof-of-Work consensus algorithm, which has not been proven to be Byzantine fault tolerant.
Storj \cite{storj} is another decentralized storage network that leverages EC for data storage. Storj uses a variation of Reed-Solomon erasure coding called forward error correction (FEC). Similar to Sia, Storj distributes data across multiple nodes to enhance fault tolerance and durability. Storj was also developed based on Ethereum, thereby utilizing Proof-of-Stake for consensus.
Both Sia and Storj demonstrate the potential of erasure coding in decentralized storage networks. %By employing redundancy through erasure coding, these networks achieve fault tolerance and data durability.

% Using EC in distributed storage 
\textbf{Decentralized Erasure Coding in Distributed Storage.}
There exist research efforts in distributed networked storage that aim to apply EC in a secure manner. Although such research efforts do not specifically address Byzantine threats, they still provide valuable insights.
Dimakis \textit{et al.} \cite{decentralizedEC} proposed decentralized erasure codes, which are created using a randomized network protocol, in which each client divides its files into chunks and sends them to randomly and independently selected storage nodes. Each storage node then creates a random linear combination of the received chunks, based on which the original files can then be recovered. Lin \textit{et al.} \cite{secureDEC} designed a secure decentralized erasure code that combines the concepts of threshold public key encryption and decentralized erasure codes. Their solution ensures that even if all storage nodes are compromised, the attacker cannot compute the content of the original file. These two schemes aim to reduce storage costs within the storage nodes by combining chunks from different files. However, such an approach has a drawback: when decoding a file, the client needs to gather data that is not directly related, thereby wasting bandwidth and computational resources.

\textbf{Analysis.}
Overall, the existing literature suggests that erasure coding can be an efficient method to ensure data availability in decentralized storage networks, while BFT consensus protocols can provide strong guarantees of data consistency and availability. 
A summary on the major adopted technologies and properties of BFT-DSN and the state-of-the-art DSNs is reported in Table~\ref{tab:related}. One can see that current DSNs do not provide verifiability when implementing EC (e.g., \cite{sia, storj}), and DSNs that employ BFT consensus either do not achieve optimal Byzantine resiliency or fail to consider the varying importance of miners who contribute different amounts of storage resources \cite{filecoin, filedag}. 

\begin{table}[htbp]
	\begin{threeparttable}
		\caption{Comparison of BFT-DSN with Existing DSNs}
			\begin{tabular}{l c c c c c}
				\toprule[1pt]
				 & \multicolumn{1}{c}{\textbf{\begin{tabular}[c]{@{}c@{}}Erasure\\ Code\end{tabular}}} & \multicolumn{1}{c}{\textbf{\begin{tabular}[c]{@{}c@{}}Verifiable \\ EC\end{tabular}}} & \multicolumn{1}{c}{\textbf{\begin{tabular}[c]{@{}c@{}}BFT\\ Consensus\end{tabular}}} & \multicolumn{1}{c}{\textbf{\begin{tabular}[c]{@{}c@{}}Byzantine\\ Threshold\end{tabular}}} \\
				\midrule[0.5pt]

    			Sia\cite{sia} & \checkmark & \xmark & \xmark & unknown \\

				Storj\cite{storj} & \checkmark & \xmark & \xmark & unknown \\

				Swarm\cite{swarm} & \xmark & \na & \xmark & unknown  \\

				Filecoin\cite{filecoin} & \xmark & \na & \checkmark & 20\% SW \\

                    FileDAG\cite{filedag} & \xmark & \na & \checkmark & 33\% \\
                    
				\textbf{BFT-DSN} & \checkmark & \checkmark & \checkmark & 33\% SW \\
				\bottomrule[1pt]
			\end{tabular}
			\label{tab:related}
		\begin{tablenotes}
		    \item SW: Storage Weighted
		\end{tablenotes}
	\end{threeparttable}
	\end{table}

\subsection{Preliminaries}

In this subsection, we provide the preliminary knowledge that are needed by our BFT-DSN design. 

\subsubsection{Decentralized Storage Network (DSN)}

DSNs aggregate storage spaces offered by multiple independent storage providers and self-coordinate to provide reliable and secure global data storage and retrieval services to clients without relying on any trusted third party. Generally speaking, the workflow of a DSN consists of two phases: put and get. Users put their files into the storage network and also get files from the network. Miners get paid by users by handling their requests. A DSN must guarantee data availability, integrity, and fault tolerance. The two common terms in a DSN that are heavily used by our BFT-DSN are sector and Proof-of-Storage (PoS). A sector is the unit of the storage space added to the network. Miners pledge sectors, contributing storage spaces to the network. These sectors act as guarantees to the network, ensuring that a specific amount of storage will stay available for a set duration for data storage. In the implementation of BFT-DSN, all sectors are standardized to a specific size, such as 32 GB.
PoS helps miners demonstrate that they have been continuously storing data in a sector. In Filecoin, Proof-of-Replication (PoRep) and Proof-of-Spacetime (PoSt) are combined to create a PoS scheme. PoRep ensures that a specified number of file copies have been generated, while PoSt confirms that a file copy has been stored continuously.
%A decentralized storage network (DSN) is a type of storage system where data is distributed across a network of nodes rather than being stored in a central location. There are two roles in a DSN: clients and miners. Clients pay tokens to use storage services, while miners earn tokens by providing services. Miners provide storage services by pledging sectors to DSN and help clients search for information about files in BFT-DSN to provide retrieval services. Transactions and proofs, which will be mentioned later, are recorded on the blockchain ledger of DSN, and all miners maintain the blockchain ledger. DSNs offer several advantages over traditional storage systems. For one, they are not controlled by any single entity, making them more resistant to censorship and surveillance. Furthermore, DSNs can be more cost-effective than traditional systems, as they can utilize unused storage space on individual devices instead of relying on expensive data centers. DSNs are increasingly being used in various applications, including file sharing, content delivery, and blockchain-based services. Here, we list some DSN-related notions that are frequently mentioned in this paper:

%\textbf{Sector.} In a DSN, storage space is added in unit of sectors and sectors are promises to the network that some storage will remain for a promised duration.
%\textbf{Host node $N_h(i)$:} The host node $N_h(i)$ refers to node that pledges the $i$th sector, $\mathsf{sec}_i$.

\subsubsection{Erasure Coding (EC)}

Erasure coding (EC) is an alternative to data replication because it incurs significantly less storage overhead while maintaining equal (or better) reliability. Many works optimize the coding efficiency and recovery bandwidth, such as pyramid coding \cite{pyramidecodes}, piggyback coding \cite{piggybackcodes}, and lazy recovery \cite{lazyrecoverycodes}. Reed-Solomon coding \cite{rscode} is the most widely used EC, in which a given file $\mathcal{F}$ is first split into $K$ data chunks and then encoded into $K+M$ chunks with $M$ parity chunks, denoted by $(K, M)$-RS. Both the data chunks and the parity chunks are termed chunks in this paper. The coding algorithm ensures that any $K$ out of the $K+M$ chunks are sufficient to reconstruct the original file, which implies that the $(K, M)$-RS can tolerate the absence of $M$ chunks. For example, in a $(4,2)$-RS coding, a 4~MB file is first divided into four 1~MB data chunks, then two additional 1~MB parity chunks are created to provide redundancy. RS coding computes parity chunks according to its data over a finite field by the following equation:
\begin{equation}
	\label{equ:vandermonde}
	\left[\begin{array}{cccc}
		1 & \alpha_0^1 & \dots & \alpha_0^{K-1} \\
		1 & \alpha_1^1 & \dots & \alpha_1^{K-1} \\
		\vdots & \vdots & \ddots & \vdots \\
		1 & \alpha_{M-1}^1 & \cdots & \alpha_{M-1}^{K-1}
	\end{array}\right] \times \left[\begin{array}{c}
	d_1 \\ d_2 \\ \vdots \\ d_K
	\end{array}\right] = \left[\begin{array}{c}
	p_{1} \\ p_{2} \\ \vdots \\ p_{M}
	\end{array}\right] ,
\end{equation}
where $d_1, \dots, d_K$ are $K$ data chunks and $p_{1}, p_{2} \dots, p_{M}$ are $M$ parity chunks. The parameters $\alpha_0, \cdots, \alpha_{M-1}$ in the Vandermonde matrix (the $M \times K$ matrix in (\ref{equ:vandermonde})) are $M$ different numbers.
%	In practice, the distributed storage system usually store all the $K+M$ chunks, and each parity chunk can be computed independently:
%	\begin{equation}
	%		p_{i\in \{1, 2, \dots, M\}}=\left[\begin{array}{c}
		%			1 \\ \alpha_{i-1}^1 \\ \vdots \\ \alpha_{i-1}^{K-1}
		%		\end{array}\right] \cdot \left[\begin{array}{c}
		%			d_1 \\ d_2 \\ \vdots \\ d_K
		%		\end{array}\right] = \sum_{j=1}^{K} \alpha_{i-1}^{j-1} d_j.
	%	\end{equation}
In this paper, we denote the process of encoding $\mathcal{D}$, a set of $K$ data chunks, into $\mathcal{C}$, a set of $K+M$ chunks, as
\begin{equation}
	\label{equ:ecencode}
	\begin{array}{c}
		\mathcal{C} = \{c_1, \dots, c_{K+M}\} =  \mathsf{Encode}(\mathcal{D}), where\\
		\mathcal{D} = \{d_1, \dots, d_{K}\},
	\end{array}
\end{equation}
and the process of recovering the $K$ data chunks in $\mathcal{D}$, given any $K$ out of the $K+M$ chunks in $\mathcal{C}$, as:
\begin{equation}
	\mathcal{D} = \mathsf{Decode}(c_{i_1}, \dots, c_{i_K}).
\end{equation}

\subsubsection{Homomorphic Fingerprints (HF)}
If a Byzantine node generates incorrect parity chunks during encoding, the encoded file may become unrecoverable. Therefore, it is necessary to verify the correspondence between a file and its erasure-coded chunks. To address this issue, Hendricks et al. \cite{verifiableEC} developed homomorphic fingerprint (HF) that enables the verification of all erasure-coded chunks. Homomorphic fingerprint $\mathsf{HF}()$ is a type of hash function with homomorphic property that the fingerprint of an encoded chunk can also be calculated by encoding the fingerprints of the data chunks. Specifically, for a $(K, M)$-RS described above, homomorphic fingerprint guarantees that if (\ref{equ:ecencode}) holds,
%\begin{equation}
%	d_i = \mathsf{encode}_i(d_1, \dots, d_K), i \in \{1, \dots, K+M\},
%\end{equation}
then
\begin{equation}
	\label{equ:hfencode}
	\{\mathsf{HF}(c_1), \dots, \mathsf{HF}(c_{K+M})\} = \mathsf{Encode}(\mathsf{HF}(d_1), \dots, \mathsf{HF}(d_K)).
\end{equation}
HF allows us to verify the correctness of $\mathsf{Encode}(d_1, \dots, d_K)$ by calculating the fingerprints of the chunks in $\mathcal{C}$ without running $\mathsf{Encode}(d_1, \dots, d_K)$. In Eq.(\ref{equ:hfencode}), $\mathsf{Encode}()$ takes $K$ fingerprints as its input, which is much smaller than $K$ chunks, the input of (\ref{equ:ecencode}). %\todo{Proof of homomorphic fingerprints is complex and introduces a lot of varieties, can't we use the fingerprinting method as a blackbox?}

\subsubsection{Weighted Threshold Signature (WTS)}
In BFT-DSN, we utilize the weighted threshold signature (WTS) \cite{weightedTS} to accelerate verification for chunks. Unlike traditional threshold signature schemes, WTS enables the designation of nodes with varying weights. Here are the WTS functions used in BFT-DSN:
\begin{enumerate}
    \item $\mathsf{WTS-Setup}(1^\kappa) \rightarrow pp$: The setup function takes the security parameter as input and outputs the public parameters $pp$ of the signature scheme.
    \item $\mathsf{WTS-KeyGen}(pp, nn, w) \rightarrow vk, ak, [{sk}_1, \dots, {sk}_{nn}]$: The key generation function takes as input the public parameters $pp$, the total number of miners $nn$, and a vector of weights $w$. The function outputs the global verification key $vk$, aggregation key $ak$, and per miner signing key ${sk}_i$.
    \item $\mathsf{WTS-PSign}(m, {sk}_i) \rightarrow \sigma_i$: Miner $i$ uses the sign function with its signing key ${sk}_i$ to generate a partial signature $\sigma_i$.
    \item $\mathsf{WTS-Aggregate}(\{\sigma_i\}, ak) \rightarrow \sigma$: On input of a set of partial signatures and the public aggregation key $ak$, the aggregate function generates an aggregate signature $\sigma$.
    \item $\mathsf{WTS-Verify}(m, \sigma, vk, t) \rightarrow 0/1$: The verify function takes as input a message $m$ and its signature $\sigma$, the global verification key $vk$, and a weight threshold $t$, and outputs 1 if and only if $m$ is signed by signers with a total weight of at least $t$.
\end{enumerate}

\subsection{Models and Goals}
%\label{sec:setting}
\label{ss:assumption}

\subsubsection{System Model} In the context of our BFT-DSN, the system comprises clients and miners. Clients use storage services by sending storage and retrieval requests through blockchain to miners. Miners provide storage resources in the form of sectors to the DSN and handle clients' requests. BFT-DSN operates on a partially synchronous network model, which incorporates a known bound $\Delta$ and an unknown Global Stabilization Time (GST). After GST, the network achieves synchronization most of the time, with all transmissions between honest nodes arriving within time $\Delta$. The scenario of BFT-DSN resembles a consortium blockchain, where Byzantine Fault Tolerance (BFT) consensus protocols are primarily employed. We refer to this scenario as a consortium DSN.

\subsubsection{Threat Model}
%\textbf{Threat Model.}
We assume a Byzantine adversary model in which dishonest nodes can carbitrarily betray the DSN protocol. Particularly the following two types of attacks are often considered in a DSN:
\begin{itemize}
	\item Sybil Attacks: Malicious miners may pretend to store more chunks (and get paid for them) than the ones actually stored by creating multiple Sybil identities. For instance, a miner may use some of its sectors to provide storage services instead of using all the sectors pledged to the network. As a result, the actual redundancy of the file may be lower than expected.
	\item Generation Attacks: Malicious miners may falsely claim to have stored a large amount of data when in reality they are generating the data on demand using a small program. This deceptive practice can increase the rewards for the miners in BFT-DSN, as rewards are based on the amount of storage being utilized. Additionally, it can compromise storage reliability by falsely claiming to offer duplicated storage when it actually not. If the size of the program is smaller than the claimed amount of stored data, the deception can be even more successful.
\end{itemize}

In this "consortium DSN", $n$ represents the total number of pledged sectors, and $f$ the maximum number of sectors controlled by all Byzantine adversaries. BFT-DSN considers Byzantine faults under the assumption that $n = 3f + 1$. It is important to note that this assumption is based on the number of sectors, not nodes, and the network scale in a consortium DSN is relatively small, typically with $n$ not exceeding 1000. We denote these $n$ sectors in the DSN as $\{\mathsf{sec}_1, \cdots, \mathsf{sec}_n\}$, and the miner that pledges $\mathsf{sec}_i$ as $N(i)$, the host miner of $\mathsf{sec}_i$.

\subsubsection{Design Goals}
%\textbf{Design Goals.}
We design BFT-DSN with the following goals:

\begin{itemize}
	\item Availability: Each file must be available for all honest nodes when the network is synchronized, which means no file can be lost if no more than $f=\lfloor\frac{n-1}{3}\rfloor$ sectors are controlled by adversary.
	\item Verifiability: Integrity of each file should be verifiable during its whole life cycle including uploading, encoding, storing, downloading and decoding, to make the owner of the file have high confidence on the integrity of the file.
	\item Efficiency: With the goals of availability and verifiability satisfied, the I/O performance of BFT-DSN, in terms of put and get latencies, should be comparable to that of the existing state-of-the-art DSNs.
\end{itemize}

\section{BFT-DSN Design}
\label{sec:design}
In this section, we begin with an overview on BFT-DSN and then detail its design.

\subsection{Overview}
\label{ss:overview}

Fig.~\ref{fig:strawman} provides an overview on BFT-DSN, a DSN that incorporates BFT erasure coding and BFT consensus for optimal resiliency. In our BFT-DSN, miners are categorized based on the functions they provide: encoding, storage, and retrieval. As shown in Fig.~\ref{fig:strawman}, an encoding miner encodes files from clients into chunks and distributes them among storage miners. A storage miner provides storage space in sectors and stores the received chunks in its sectors. A retrieval miner collects chunks from storage miners, decodes them, and then sends the recovered file to the requesting client.

% operations
More specifically, a file is first encoded into chunks by an encoding miner, then each chunk is stored in a sector provided by a storage miner. To guarantee security against Byzantine attackers, BFT-DSN makes use of homomorphic fingerprints and threshold signatures to verify the correctness of the encoding process and the integrity of the chunks. These techniques also ensure that the integrity of the chunks collected during get operations can be effectively verified.
BFT-DSN employs blockchain to audit events such as put/get operations and expiration of the storage for a file. It utilizes a novel storage-weighted BFT consensus algorithm to reach consensus on new blocks added to the blockchain. This consensus algorithm takes into account the differences in the amounts of storage resources contributed by the storage miners and achieves optimal Byzantine resiliency. To monitor the storage space physically contributed by each miner and support consensus, we develop a Proof of Storage scheme for BFT-DSN, which uses a Merkle tree.
% Byzantine resiliency
As an example to illustrate the Byzantine resilience of BFT-DSN, Fig.~\ref{fig:strawman} illustrates that, even if storage miner 1, who controls one sector, is compromised and refuses to provide the chunk it stores, the retrieval miner is still able to decode the requested file from the collected chunks.

\begin{figure}[!htbp]
    \centering
    \includegraphics[width=0.48\textwidth]{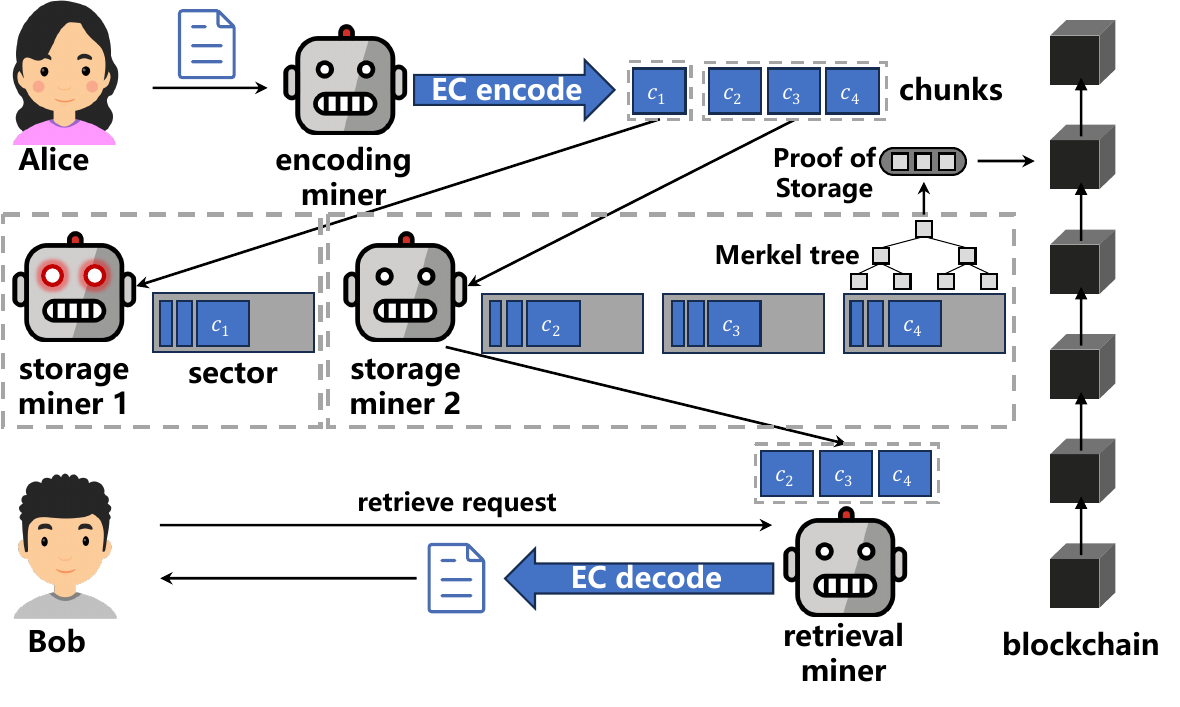}
    \caption{An overview on BFT-DSN}
    \label{fig:strawman}
\end{figure}

\subsection{The Put Operation in BFT-DSN}
\label{ss:storage}

\begin{algorithm}[htbp]
	\caption{Put Operation in BFT-DSN}
	\label{alg:put}
	
	\SetKwProg{Fn}{Function}{}{end}
	\SetKwProg{Upon}{upon}{ do}{end}
	\SetKwFunction{verifychunk}{$\mathsf{VerifyChunk}$}
	\SetKwFunction{verifyHF}{$\mathsf{VerifyID}$}
	\SetKwFunction{signHF}{$\mathsf{SignFingerprints}$}
	\SetKwFunction{split}{$\mathsf{Split}$}
	\SetKwFunction{encode}{$\mathsf{Encode}$}
	\SetKwFunction{decode}{$\mathsf{Decode}$}
	\SetKwFunction{send}{$\mathsf{Send}$}
	\SetKwFunction{tssign}{$\mathsf{WTS\text{-}PSign}$}
	\SetKwFunction{tsaggre}{$\mathsf{WTS\text{-}Aggregate}$}
	\SetKwFunction{tsverify}{$\mathsf{WTS\text{-}Verify}$}
	\SetKw{fileF}{$\mathcal{F}$}
	\SetKw{ID}{$\mathsf{ID}$}
	\SetKw{TX}{$\mathsf{TX}$}
	\SetKw{sec}{$\mathsf{sec}$}
	\SetKwFunction{hf}{$\mathsf{HF}$}

        % client
        \role{Client} \\
        send $\mathcal{F}$ to an encoding miner \\
        $\mathcal{D} = \{d_1, \cdots, d_{n-f}\} \leftarrow \split(\fileF)$\\ 
        $\mathcal{H} = \{h_1, \cdots, h_{n-f}\} \leftarrow \{\hf(d_1), \cdots, \hf(d_{n-f})\}$ \\
        $\ID_{\fileF} \leftarrow H(h_1\|\cdots\|h_{n-f})$ \\
	$\TX_{\fileF} \leftarrow \langle \mathsf{STORE}, \ID_{\fileF}, \mathcal{H} \rangle$ \\	
        send $\TX_{\fileF}$ to blockchain \\
        % encoding miner
	\role{Encoding Miner} \\
	\textbf{Input:} a file \fileF \\
        \textbf{Output:} encoded chunks $\mathcal{C}$ \\
	$\mathcal{D} = \{d_1, \cdots, d_{n-f}\} \leftarrow \split(\fileF)$\\ 
	$\mathcal{C} = \{c_1, \cdots, c_n\} \leftarrow \encode(\mathcal{D})$ \\
	\For{$i$ \textbf{from} $1$ \textbf{to} $n$}{
		send $c_i$ to $N(i)$
	}
	% follower
	\role{Storage Miner $p$}\\
	\Upon{receiving $\TX_{\fileF}$}{
        // denote $\TX_{\fileF}.\mathcal{H}$ by $\{h_1, h_2, \cdots, h_{n-f}\}$ \\
		\If{$\TX_{\fileF}.\ID_{\fileF} \stackrel{?}{=} H(h_1\|\cdots\|h_{n-f})$}{
			%\signHF($TX_{\fileF}.hf[]$)
			$\{h_1', \dots, h_n'\} \leftarrow \encode(\TX_{\fileF}.\mathcal{H})$ \\
			\For{$i$ \textbf{from} $1$ \textbf{to} $n$}{
				$\sigma_p \leftarrow \tssign(h_i', {sk}_p)$\\
				send $\sigma_p$ to $N(i)$
			}
		}
	}
	\Upon{receiving signatures for $\hf(c)$ with total weight $\ge f+1$}{
		// denote received signatures for $\hf(c)$ by $\{\sigma_i\}$\\
		$\sigma^{(c)} \leftarrow$\tsaggre($\{\sigma_i\}, ak$)\label{codeline:aggresig}\\
		store $\sigma^{(c)}$
	}
	\Upon{receiving chunk $c$}{
            \If{\tsverify($\hf(c), \sigma^{(c)}, vk, f+1$) \textbf{is} true}{store $c$}
		%\verifychunk($c_i$)\label{codeline:verifychunk}
	}
\end{algorithm}

In this section, we describe the put operation of BFT-DSN, where a client puts a file $\mathcal{F}$ onto the network. The pseudocode of this operation is shown in Algorithm~\ref{alg:put}.

First, the client sends file $\mathcal{F}$ to an encoding miner to encode the file into chunks (line 2). The encoding miner is randomly selected from all miners based on their storage weights (explained in Section~\ref{ss:consensus}). To ensure that miners who receive these chunks can verify their integrity, the client signs a $\mathsf{STORE}$ transaction to the blockchain and includes the homomorphic fingerprints of the data chunks in that transaction (line 3-7). Homomorphic fingerprints of the parity chunks are not included to save on-chain storage as they can be calculated from those of the data chunks using Eq.~(\ref{equ:hfencode}). Additionally, to establish a binding relationship between these homomorphic fingerprints and the file, we set the identifier of the file $\mathcal{F}$ to be $\mathsf{ID}_\mathcal{F}$, which is the hash of the concatenation of the homomorphic fingerprints (line 5), and include it in the $\mathsf{STORE}$ transaction (line 6).  %$\mathsf{ID}_\mathcal{F}$ is also included in the $\mathsf{STORE}$ transaction.

When receiving a file from a client, an encoding miner encodes the file into $n$ chunks using a $(K, M)$-RS code (line 11-12). Here $K=n-f$ and $M=f$, which is to be explained later. Then the encoding process of the file $\mathcal{F}$ can be represented as follows (line 11-12):
\begin{equation}
	\mathcal{D} = \{d_1, \cdots, d_{n-f}\} = \mathsf{Split}(\mathcal{F}),
\end{equation}
\begin{equation}
	\label{equ:encode}
	\mathcal{C} = \{c_1, \cdots, c_n\} = \mathsf{Encode}(\mathcal{D})
\end{equation}
After encoding $\mathcal{F}$ into $n$ chunks, the encoding miner distributes them to the $n$ sectors in BFT-DSN, one chunk to each sector (line 13-14). Particularly, $c_i$ is forwarded to $\mathsf{sec}_i$, where $i = 1, \cdots, n$. Note that one miner can pledge multiple sectors, therefore different chunks may be received by the same miner.

When receiving a $\mathsf{STORE}$ transaction from the blockchain, a storage miner first verify the correspondence between the identifier $\mathsf{ID}_\mathcal{F}$ and the homomorphic fingerprints of the data chunks. This can be done by recalculating the hash of the concatenation of the homomorphic fingerprints of the data chunks and ensuring that the resulting hash equals $\mathsf{ID}_\mathcal{F}$ (line 18). If the verification succeeds, the miner proceeds to calculate the homomorphic fingerprints of the parity chunks. This is done by EC encoding the homomorphic fingerprints of the data chunks, leveraging the homomorphic property (line 19). Next the miner signs the homomorphic fingerprint of each chunk with its signing key, generates a partial signature (line 21), then sends it to the miner storing the corresponding chunk (line 22).

To expedite the verification process in BFT-DSN, we assign each miner a weight that is equal to the number of sectors they pledge. By this way, each miner can receive partial signatures with a total weight at least $n-f$ for each chunk it stores. These signatures are then combined to generate an aggregate signature for that chunk (line 23-26). As a result, the storage miner can verify the integrity of the chunk received from the encoding miner using the corresponding aggregate signature. If the verification succeeds, the storage miner stores the chunk (lines 27-29).

%Note that the above process is proposed for verifying the integrity of the parity chunks, i.e., the integrity of the EC encoding process (Question 1). After receiving the partial signatures of the HF of a particular chunk, a storage minder can verify its correctness via 

\textbf{Choosing the EC parameters $K$ and $M$.}
Here we explain how we determine the values of $K$ and $M$, the two parameters of RS-coding. We use $n$ to represent the total number of sectors in the BFT-DSN network, and $f$ the maximum number of Byzantine sectors that can be tolerated. The value of $f$ is calculated as $f=\lfloor\frac{n-1}{3}\rfloor$. Since we distribute each chunk of a file to one sector, the total number of chunks encoded from a file is equal to the total number of sectors, i.e., $K+M=n$. Next we deduce the relationship of $K$, $M$, and $f$.
First, it is important to consider generation attacks in this context. With a $(K, M)$-RS code, if an adversary possesses $K$ or more chunks of the file $\mathcal{F}$, it can pretend to have stored any number of chunks of $\mathcal{F}$. This is possible because the adversary can recover $\mathcal{F}$ using any $K$ different chunks and encode additional chunks as needed. To protect against such an attack, $K$ must be set to a value explicitly greater than the maximum number of sectors controlled by Byzantine adversaries, which is $f$. 
Second, with the existence of adversaries hosting $f$ Byzantine sectors, at most $f$ chunks may be dropped or tampered. As the integrity of each chunk is verified with HF and WTS, tampered chunks can be identified. Thus, we need to make the file recoverable when $n-f$ correct chunks are available. Consequently, $K$ must be equal to or smaller than $n-f$.
Last but not least, it can be seen that the storage cost of a file $\mathcal{F}$ is $\frac{n}{K}\cdot |\mathcal{F}|$ (see Section~\ref{sec:analysis}). This indicates that the bigger the $K$, the lower the storage cost.
As a result, we let $K=n-f$, the largest value in the interval $(f, n-f]$. This interval is obviously not empty under the assumption that $n = 3f+1$. Thus we have $M=n-K=f$. As a conclusion, $(n-f, f)$-RS coding is chosen.

\textbf{The effect of HF and WTS in BFT-DSN.}
The calculated homomorphic fingerprints in line 19 are used to verify the integrity of the parity chunks. If the parity chunks are correctly encoded from the data chunks of $\mathcal{F}$, their homomorphic fingerprints would match the calculated homomorphic fingerprints. By using HF to verify the integrity of each chunk, BFT-DSN guarantees the security of its EC-based decentralized storage. %a storage miner no longer needs to collect the other chunks, recover the file, and then encode the file to obtain the correct chunk for verification. This significantly reduces the bandwidth usage and computation overhead during the verification process.
Additionally, to expedite the verification process, we adopt the weighted threshold signature scheme, %We assign each miner a signing weight equal to the number of sectors they pledge. The number of sectors is validated by Proof-of-Storage, which is discussed in Section~\ref{sss:pos}. Once the expected parity chunks' homomorphic fingerprints are calculated, the miner $p$ signs each calculated homomorphic fingerprint $h_i'$ with its signing key, generating a partial signature $\sigma_p$ (line 21). This partial signature is then sent to the miner storing the corresponding chunk (line 22). 
to ensure that the verification process can be done without the need of gathering and encoding $n-f$ homomorphic fingerprints every time. Note that while one can guarantee that each correct homomorphic fingerprint is signed by the miners with a total weight of at least $n-f$, an aggregate signature with a weight of $f+1$ is sufficient to validate a homomorphic fingerprint. This is because no honest miner would sign an invalid homomorphic fingerprint, and the sign weight of Byzantine miners does not exceed $f$. Therefore, we set the weight threshold for stopping waiting for partial signatures (line 23) and passing $\mathsf{WTS\text{-}Verify}$ (line 28) to be $f+1$ to avoid unnecessary latency.

\textbf{Storage Cost.}
Finally we briefly discuss the storage cost of BFT-DSN.
To store a file $\mathcal{F}$, $(n-f, f)$-RS first splits $\mathcal{F}$ into $n-f$ data chunks, each with a size of $|\mathcal{F}|/(n-f)$, then generates $f$ parity chunks of the same size. The total storage cost for the $n$ chunks is $n \cdot |\mathcal{F}|/(n-f)$.
Based on the assumption that $n=3f+1$, the storage cost of file $\mathcal{F}$, namely $\mathsf{C}_{\mathcal{F}}$, is
\begin{equation}
	\mathsf{C}_{\mathcal{F}} = \frac{3f+1}{2f+1}\cdot |\mathcal{F}| < \frac{3}{2}\cdot |\mathcal{F}| .
\end{equation}
Therefore, $\frac{3}{2}\cdot |\mathcal{F}|$ is an upperbound of $\mathsf{C}_\mathcal{F}$. 

\subsection{Consensus of BFT-DSN}
\label{ss:consensus}

In order to practically optimize the resiliency of BFT consensus in DSN, we design the Storage-Weighted BFT (SW-BFT) consensus, which is based on the Proof-of-Storage algorithm and Tendermint Core. In this section, we provide a detailed explanation on our carefully designed Proof-of-Storage algorithm, followed by a formal description on the SW-BFT consensus algorithm.

\subsubsection{Proof of Storage in BFT-DSN}
\label{sss:pos}

In our assumption where $n = 3f+1$, we define $n$ and $f$ based on the number of sectors rather than nodes. Specifically, $n$ represents the number of sectors pledged in BFT-DSN, while $f$ denotes the maximum number of sectors that Byzantine adversaries can control. Hence, to uphold our SW-BFT consensus design, it's crucial to verify the authenticity of each pledged sector.
In our BFT-DSN, PoSs are used to confirm whether a host miner has continuously stored data in a sector. These proofs also help adjust the weights of the miners during the consensus process.
Unlike Filecoin, which uses the computationally heavy PoRep algorithm to ensure storage replication and duplication, BFT-DSN achieves storage duplication through EC. As a result, the PoS algorithm in BFT-DSN focuses on achieving Proof-of-Retrievability \cite{por}. Another observation is that, as described in Section~\ref{ss:storage}, each time a file is put on BFT-DSN, the total sizes of the chunks in any two sectors are the same. This is because for every file put, equally sized chunks are distributed across all sectors, with one chunk going to each sector. Besides, in BFT-DSN, all sectors have the same size, such as 32 GB. %Therefore, the positions of the sectors in BFT-DSN are equivalent.
Based on these observations, we design the PoS scheme of BFT-DSN as shown in Fig.~\ref{fig:pos}:

\begin{figure*}[bthp]
    \centering
    \includegraphics[width=\textwidth]{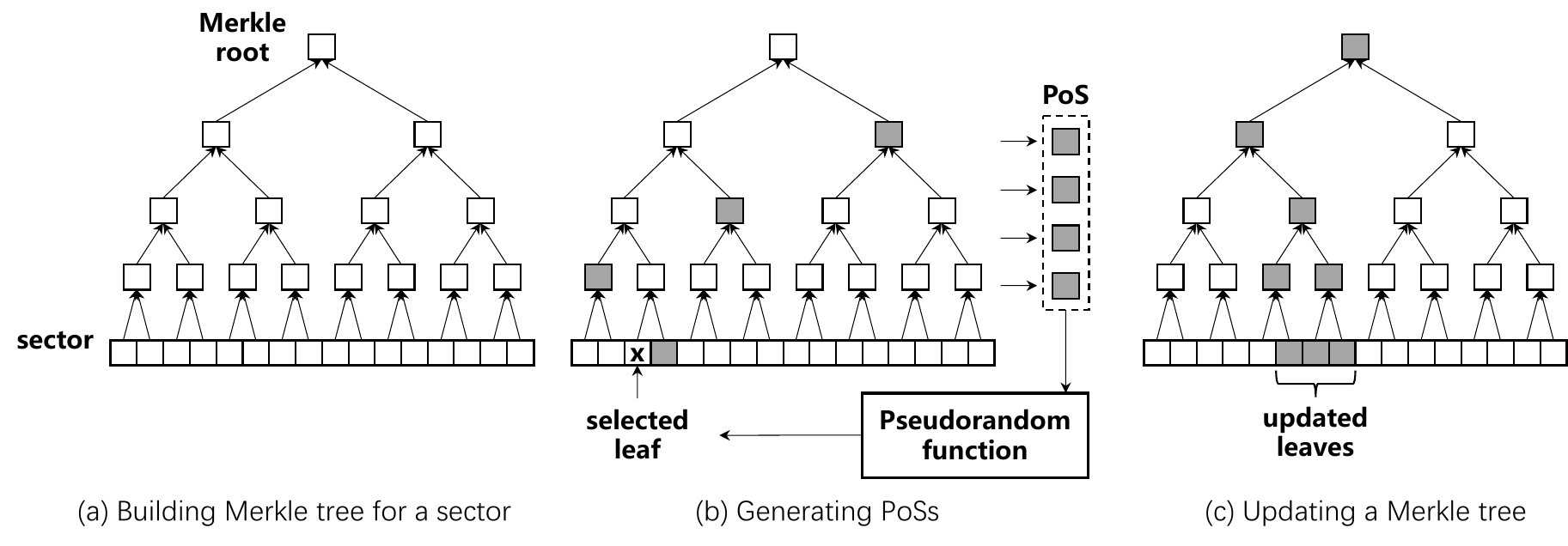}
    \caption{The PoS scheme in BFT-DSN}
    \label{fig:pos}
\end{figure*}

% \begin{figure*}[tbhp]
% 	\centering
% 	\begin{subfigure}{5.5cm}
% 		\includegraphics[width=\textwidth]{images/posa.pdf}
% 		\caption{Building Merkle tree for a sector}
% 		\label{fig:storagecosts}
% 	\end{subfigure}
% 	\qquad
% 	\begin{subfigure}{5.5cm}
% 		\includegraphics[width=\textwidth]{images/posb.pdf}
% 		\caption{Generating PoSs}
% 		\label{fig:storagecosts_per_mb}
% 	\end{subfigure}
%         \qquad
% 	\begin{subfigure}{5.5cm}
% 		\includegraphics[width=\textwidth]{images/posc.pdf}
% 		\caption{Updating a Merkle tree}
% 		\label{fig:storagecosts_per_mb}
% 	\end{subfigure}
% 	\caption{PoS algorithms}
% 	\label{fig:pos}
% \end{figure*}

\textbf{Building A Merkle Tree.} The core concept of PoS in BFT-DSN involves employing a Merkle tree, which is a hash tree constructed from the current data in a sector. When a sector is pledged, it is first filled with random data, and a Merkle tree is also built from it. Each leaf in the Merkle tree represents a data fragment of certain size, e.g., 256 B. The hashes of adjacent leaf nodes are combined and hashed to create parent nodes. All the nodes, either a leaf node, or an internal node, have the same size. This process continues until there is only one root node left, representing the entire sector. After a Merkle tree of a sector is built, the Merkle root is sent to the blockchain to support PoS verification.

\textbf{Generating a Proof.} To produce a PoS for a sector, its host miner picks a random leaf from the sector's Merkle tree. The miner then supplies hashes corresponding to the siblings of each node along the path from the chosen leaf to the Merkle tree's root. Using the provided hashes, a verifier rebuilds the tree branches leading up to the root node and checks if the root node matches the expected value stored on the blockchain. Each PoS is broadcast as a transaction to the blockchain, allowing all nodes to verify it. A newly pledged sector is considered committed once its initial PoS is documented on the blockchain.

\textbf{Verifying a Proof.}
If a miner stores data correctly, the hashes it provides will be valid. This process enables the verifier to reconstruct the tree path and confirm the root node. However, if the storage is incorrect, the hashes won't match, which allows the verifier to identify the problem.
Miners generate PoSs periodically. To generate a new PoS for a sector, the host miner selects another leaf. This selection is based on a pseudorandom function that uses the previous PoS as input. The miner then repeats the steps described above.

\textbf{Updating the Merkle Tree.} When a set of chunks in a sector changes (for example, when a new file is added to BFT-DSN or a file in BFT-DSN expires), the Merkle tree built from the sector must be updated. When adding or removing a chunk to or from a sector, it means modifying a continuous segment of data within the updated sector. As a result, the corresponding leaves will be recalculated along with their ancestors. All other nodes in the Merkle tree remain unchanged. Once the Merkle tree is updated, the new Merkle root is sent to the blockchain. Generation of subsequent PoSs will be based on this updated Merkle tree.

\subsubsection{Storage-Weighted BFT (SW-BFT) Consensus}

To enhance the resiliency of BFT-DSN, we introduce a BFT consensus algorithm tailored for DSNs. This algorithm takes into account the number of sectors pledged by miners and adjusts their consensus weight accordingly. We call this the Storage-Weighted BFT (SW-BFT) consensus, which combines the PoS scheme of BFT-DSN with Tendermint Core, a BFT protocol that facilitates consensus among the miners with differing weights \cite{tendermint}. 
Note that to monitor the number of sectors each miner pledges and ensure consistent monitoring across all nodes in SW-BFT, we let each miner maintain a weight table that tracks the numbers of pledged sectors of other miners over time. The weight table is updated for new sector pledges (which increments the consensus weight), failed PoSs (which decrements the consensus weight), and other storage and consensus faults. 

SW-BFT reuses the three voting phases of the Tendermint Core to reach consensus on a block: prevote, precommit, and commit. Each phase necessitates $\frac{2}{3}$ of the votes to proceed. The key distinction between SW-BFT and the Tendermint Core lies in SW-BFT checking the $\frac{2}{3}$ vote threshold (corresponding to $n-f$ sectors out of $n$, where $n = 3f+1$) based on the weight table.
For instance, during the prevote phase, each miner collects prevote messages from other miners. These messages contain the block's hash being voted on and are signed by the voters. Upon receipt of a prevote message, a miner checks the voter's weight from its weight table, then adds this weight to the total weight of the votes received for the block. If this total weight reaches $n-f$, %of the total weight of all miners in the BFT-DSN, 
the miner moves to the next phase.
The weight accumulating procedures for the other two voting phases, namely precommit and commit, are similar to the prevote phase.

It's worth noting that every event updating the weight table corresponds to an on-chain transaction. Such an event becomes effective when the related transaction is confirmed on the blockchain. During a block's consensus process, each miner's view of the blockchain is locked at the previous block height. This ensures that the weight table is consistently observed by each miner in the consensus. Consistency in monitoring the weight table is crucial, as any inconsistency could lead to a consensus failure \cite{tendermint}.

\subsection{The Get Operation in BFT-DSN}

In this section, we explain the get operation in BFT-DSN, in which a client retrieves a file $\mathcal{F}$ from BFT-DSN using $\mathsf{ID}_\mathcal{F}$, the identifier of $\mathcal{F}$. The pseudocode for the get operation is shown in Algorithm~\ref{alg:decode}.

\begin{algorithm}[!htbp]
	\caption{The Get Operation in BFT-DSN}
	\label{alg:decode}

        \SetKwProg{Fn}{Function}{}{end}
	\SetKwProg{Upon}{upon}{ do}{end}
	\SetKwFunction{verifychunk}{$\mathsf{VerifyChunk}$}
	\SetKwFunction{verifyHF}{$\mathsf{HF\text{-}Verify}$}
	\SetKwFunction{signHF}{$\mathsf{SignFingerprints}$}
	\SetKwFunction{split}{$\mathsf{Split}$}
	\SetKwFunction{encode}{$\mathsf{Encode}$}
	\SetKwFunction{decode}{$\mathsf{Decode}$}
	\SetKwFunction{send}{$\mathsf{Send}$}
	\SetKwFunction{tssign}{$\mathsf{WTS\text{-}Sign}$}
	\SetKwFunction{tsaggre}{$\mathsf{WTS\text{-}Aggregate}$}
	\SetKwFunction{tsverify}{$\mathsf{WTS\text{-}Verify}$}
	\SetKw{fileF}{$\mathcal{F}$}
	\SetKw{ID}{$\mathsf{ID}$}
	\SetKw{TX}{$\mathsf{TX}$}
	\SetKw{sec}{$\mathsf{sec}$}
    \SetKw{msg}{$\mathsf{msg}$}
    \SetKw{req}{$\mathsf{req}$}
	\SetKw{ts}{$\mathsf{wts}$}
	\SetKwFunction{join}{$\mathsf{Concatenate}$}
	\SetKw{ccount}{$\mathsf{chunk\_count}$}
	\SetKwFunction{hf}{$\mathsf{HF}$}
	
	\role{Client}\\
	$\req_{\fileF} \leftarrow \langle \mathsf{RETRIEVE}, \ID_{\fileF} \rangle$\\
	send $\req_{\fileF}$ to a retrieval miner\\
	
	\role{Retrieval Miner}\\
	\textbf{Input:} $\req_{\fileF}$\\
	\textbf{Output:} $\fileF$\\
	broadcast $\req_{\fileF}$\\
	$\ccount \leftarrow 0$\\
	$\mathcal{C} \leftarrow$ an empty set\\
	\Upon{receiving $(c, \sigma^{(c)})$}{
		\If{\tsverify($\hf(c), \sigma^{(c)}, vk, f+1$)}{
			$\mathcal{C} \leftarrow \mathcal{C} \cup \{c\}$
		}
		\If{$|\mathcal{C}| \ge n-f$}{
			$\{d_1, \dots, d_{n-f}\} \leftarrow \decode(\mathcal{C})$\\
			$\fileF \leftarrow \join(d_1, \dots, d_{n-f})$\\
			\Return{$\fileF$}
		}
	}
	
	\role{Storage Miner $p$}\\
	\Upon{receiving $\req_{\fileF}$ from a retrieval miner}{
            \For{each chunk $c$ of \fileF stored by $p$}{
            send $(c, \sigma^{(c)})$ to the retrieval miner %// $\sigma_i$ is the aggregate $\mathsf{WTS}$ of $c_i$
            }
		% ${\msg}_i \leftarrow \langle \mathsf{CHUNK}, c_i, \ts_i \rangle$\\
		% \send(${\msg}_i, \TX_{\fileF}.q$)
	}
\end{algorithm}

First, the client randomly selects a retrieval miner and sends a retrieval request, identified by $\mathsf{ID}_{\mathcal{F}}$, to the retrieval miner (line 2-3). Like a storage miner in the put operation, the retrieval miner is chosen from all miners based on their storage weight. When the retrieval miner receives the request, it broadcasts it to all storage miners (line 7). After receiving the request, each storage miner sends the chunks of $\mathcal{F}$ and the accompanied aggregate signatures it stores to the retrieval miner (line 18-20). This allows the retrieval miner to verify the integrity of each chunk (line 11). Once the retrieval miner collects $n-f$ correct chunks, it can decode $\mathcal{F}$ (line 13-16):
\begin{equation}
	\{d_1, \dots, d_{n-f}\} = \mathsf{Decode}(c_{i_1}, \dots, c_{i_{n-f}}) ,
\end{equation}
\begin{equation}
	\mathcal{F} = \mathsf{Concatenate}(d_1, \dots, d_{n-f}) .
\end{equation}

To manage the scenario where the selected retrieval miner is Byzantine, we establish a timeout for the retrieval process. It's evident that a retrieval miner can always receive at least $n-f$ correct chunks for a file $\mathcal{F}$ and decode it within a fixed time. This is because there are at least $n-f$ sectors run by honest miners, each storing a chunk of $\mathcal{F}$. If the retrieval miner fails to produce the requested file $\mathcal{F}$ within a specific time interval, it is deemed faulty. Any node can report this, and a different retrieval miner is then selected.
If the retrieval miner outputs a file and sends it to the client in time, the client then verifies the integrity of the retrieved file, $\mathcal{F}$, using $\mathsf{ID}_{\mathcal{F}}$. If the selected retrieval miner sends an incorrect file, the client reports the error and attempts retrieval with a different retrieval miner.

\section{Byzantine Fault Tolerance of BFT-DSN}
\label{sec:analysis}

In this section, we analyze the key property of BFT-DSN: Byzantine fault tolerance, in which the fundamental concern is to ensure file retrievability even in the presence of Byzantine adversaries. This is achieved by ensuring the safety and liveness of consensus and the verifiability throughout a file's lifecycle.

\textbf{Safety and Liveness of SW-BFT Consensus.} First, we demonstrate that the SW-BFT consensus achieves both safety and liveness, even when dealing with Byzantine adversaries controlling up to $\lfloor\frac{n-1}{3}\rfloor$ out of $n$ sectors.

\begin{theorem} % Soundness of PoS defends Sybil attacks
    (Soundness of the PoS algorithm in BFT-DSN). A miner can generate a valid PoS for a sector if and only if it is storing the data that belongs to that sector.
\end{theorem}

\begin{proof}
If a miner is storing all data in a sector, the miner can generate a PoS for this sector using the proof generation algorithm described in Section~\ref{ss:consensus}. The generated PoS can be validated because the Merkle root on the blockchain corresponds to the root of the Merkle tree constructed from all the data in the sector.
If a miner is not storing all the data in a sector, meaning that it has erased or modified some portion of the data in the sector, the corresponding leaves in the Merkle tree of that sector would be missing or changed. If one of these leaves is selected during the generation of a PoS, the miner will be unable to provide a Merkle path that produces the Merkle root stored on the blockchain. As  Merkle paths are continuously queried, sooner or later the erasure or modification of data can be detected through PoS verification.
\end{proof}

The soundness of the PoS algorithm in BFT-DSN prevents Byzantine adversaries from deploying Sybil attacks, in which a Sybil attacker pretends to maintain an arbitrary number of sectors and breaks the assumption that no more than $\lfloor\frac{n-1}{3}\rfloor$ sectors are controlled by Byzantine adversaries.

% Now we prove that assuming no more than $f=\lfloor\frac{n-1}{3}\rfloor$ sectors out of all $n$ sectors in BFT-DSN are controlled by Byzantine adversaries, SW-BFT guarantees safety and liveness.

\begin{theorem}
    With no more than $\lfloor\frac{n-1}{3}\rfloor$ out of $n$ sectors controlled by Byzantine adversaries, SW-BFT guarantees safety and liveness as defined below:
    \begin{itemize}
        \item Safety: If an honest miner commits block $B$ at height $h$, no other honest miner decides on any block other than $B$ at height $h$.
        \item Liveness: Consensus on any block eventually ends.
    \end{itemize}
\end{theorem}

\begin{proof}
    Based on the soundness of the PoS algorithm in BFT-DSN, a miner can generate a valid PoS for a sector iff it is actively maintaining that sector. As a result, a sector is included in the weight table iff it is being maintained by a miner. Therefore, if we assume that Byzantine adversaries maintain at most $f$ sectors and a total of $n = 3f+1$ sectors are being maintained, Byzantine miners should have less than $\lfloor\frac{n-1}{3}\rfloor$ voting power. It has been proven that if there is no more than $\lfloor\frac{n-1}{3}\rfloor$ Byzantine voting power out of $n$, Tendermint Core guarantees safety and liveness \cite{tendermintwhitepaper}. Thus, the SW-BFT consensus guarantees safety and liveness.
\end{proof}

\textbf{Verifiability.} Secondly, we analyze the verifiability in BFT-DSN. We aim to demonstrate that during the entire lifecycle of a file, which includes the stages of uploading, encoding, storing, downloading, and decoding, it is impossible for a Byzantine miner to output a wrong file or chunk without being detected. This holds true even when $\lfloor\frac{n-1}{3}\rfloor$ out of $n$ sectors are controlled by Byzantine adversaries.

During file uploading, when a client sends the file $\mathcal{F}$ to an encoding miner, the integrity of the received file $\mathcal{F}'$ can be verified by recalculating $\mathsf{ID}_{\mathcal{F}'}$ and comparing it with $\mathsf{ID}_\mathcal{F}$ stored on the blockchain.
During EC encoding, an encoding miner generates $n$ chunks from the file $\mathcal{F}$. The correctness of the encoding process, which ensures that the $n$ chunks are derived from $\mathcal{F}$ without any tampering, is verified using homomorphic fingerprints. This verification relies on the homomorphic property, as shown in Eq.~\ref{equ:hfencode}.
During storage, the integrity of each chunk can be verified with PoS. BFT-DSN does not generate PoS for each chunk, but instead, it generates PoS for each sector. This is sufficient, as the integrity of a sector is a necessary condition for the integrity of a chunk within it.
During chunk downloading, miners send chunks along with the aggregate threshold signatures to the retrieval miner. The integrity of each chunk can be verified using its aggregate threshold signature. An aggregate threshold signature can pass verification only if the miners with a total weight of at least $f+1$ verify its homomorphic fingerprints and sign it. Additionally, Byzantine adversaries cannot have a weight higher than $f$ based on the weight assignment of the weighted threshold signature and the soundness of PoS in BFT-DSN. As a result, the verification of the progress of chunk downloading is achieved.
Finally, a retrieval miner decodes the collected chunks and sends the output file to the client. The verification of the integrity of the received file $\mathcal{F}'$ is performed in the same way as that during file uploading. It is also based on $\mathsf{ID}_\mathcal{F}$.

\textbf{Retrievability.} Finally, we demonstrate the Byzantine fault tolerance of BFT-DSN by showing that it satisfies the fundamental property of a DSN: retrievability, even in the presence of Byzantine adversaries. In the following we formally prove that BFT-DSN achieves retrievability even when up to $\lfloor\frac{n-1}{3}\rfloor$ sectors are controlled by Byzantine adversaries.

\begin{lemma}
	\label{lemma:honest_retrieval_miner}
When there are a total of $n$ sectors in BFT-DSN, if at most $\lfloor\frac{n-1}{3}\rfloor$ sectors are controlled by Byzantine adversaries, an honest retrieval miner can successfully output the file $\mathcal{F}$ during the retrieval process.
\end{lemma}

\begin{proof}
During the retrieval process of any file $\mathcal{F}$, an honest retrieval miner can obtain at least $n-f$ different correct chunks of $\mathcal{F}$. This is because during the storage process of $\mathcal{F}$, the $n$ chunks are distributed to $n$ sectors, of which at least $n-f$ are honest. Therefore, at least $n-f$ chunks of $\mathcal{F}$ are stored in honest sectors. Sectors controlled by Byzantine nodes may provide incorrect chunks, which can be detected and rejected by an honest retrieval miner based on verifiability. Thus, an honest retrieval miner can always successfully gather $n-f$ correct chunks, which are sufficient to decode $\mathcal{F}$.
\end{proof}

\begin{theorem}
When there are a total of $n$ sectors pledged in BFT-DSN, if no more than $\lfloor\frac{n-1}{3}\rfloor$ sectors are controlled by Byzantine adversaries, BFT-DSN guarantees file retrievability. This means that for any file $\mathcal{F}$, its retrieval process will eventually succeed and finish in an expected number of $O(1)$ tries.
\end{theorem}

\begin{proof}
According to Lemma \ref{lemma:honest_retrieval_miner}, an honest retrieval miner can always successfully output $\mathcal{F}$. If the retrieval miner is faulty and refuses to provide the correct file, a timeout will be triggered. In this case, the client can switch to another retrieval miner. Since the retrieval miner is randomly chosen based on the number of sectors it controls, and according to our assumption that the total number of sectors controlled by Byzantine miners is at most $\lfloor\frac{n-1}{3}\rfloor$ out of $n$, the probability that the chosen retrieval miner appears to be Byzantine, denoted as $\mathsf{Pr_{\mathsf{Byzantine}}}$, is less than $1/3$. Therefore, the expected number of tries needed to find an honest retrieval miner is
\[\frac{1}{1-\mathsf{Pr_{\mathsf{Byzantine}}}} < \frac{1}{1-\frac{1}{3}} = \frac{3}{2} = O(1).\]
\end{proof}

% \subsection{Storage Cost}

% To store a file $\mathcal{F}$, $(n-f, f)$-RS first splits $\mathcal{F}$ into $n-f$ data chunks that has size $|\mathcal{F}|/(n-f)$, then generates $f$ parity chunks of same size. The $n$ chunks takes $n\cdot |\mathcal{F}|/(n-f)$ storage costs in total. Based on the assumption that $n=3f+1$, the storage cost of file $\mathcal{F}$, namely $\mathsf{C}_{\mathcal{F}}$ is
% \begin{equation}
% 	\mathsf{C}_{\mathcal{F}} = \frac{3f+1}{2f+1}\cdot |\mathcal{F}| < \frac{3}{2}\cdot |\mathcal{F}| .
% \end{equation}
% Therefore, $\frac{3}{2}\cdot |\mathcal{F}|$ is an upperbound of $\mathsf{C}_\mathcal{F}$. 

% Although existing DSNs don't consider Byzantine fault tolerence, they can tolerate Byzantine faults by duplicated storing files. Under our assumption, for any file $\mathcal{F}$, current DSNs can store $f+1$ copies of $\mathcal{F}$ to tolerate $f$ Byzantine faults. Then the storage cost of an existing DSN to tolerate the same kind of Byzantine faults is $(f+1)\cdot |\mathcal{F}|$. In most permissioned deployments, $f+1$ is far bigger than $\frac{3}{2}$. Hence, the storage cost of storing a file in BFT-DSN, which tolerates $f$ storage faults, is significantly lower compared to storing the same file in a replication-based DSN with the same number of storage faults.

\section{Experiments}
\label{sec:exp}

To evaluate the performance of BFT-DSN, we conducted real experiments. Specifically, we implemented BFT-DSN based on the description in Section~\ref{sec:design}, and tested put and get operations.

\subsection{Implementation}

%We implement BFT-DSN by combining Filecoin Lotus and Tendermint Core. Filecoin Lotus is the most popular implementation of Filecoin, which is the biggest DSN\footnote{https://spec.filecoin.io/implementations}. Tendermint Core is an open-source BFT consensus engine.
Fig.~\ref{fig:diagram} shows the architecture of our BFT-DSN implementation. Throughout the life cycle of a file, its chunks are transmitted from the encoding miner to the storage miners, and then to the retrieval miner. The encoding miner generates the chunks, while the storage miners generate homomorphic fingerprints for these chunks and create weighted threshold signatures based on the homomorphic fingerprints. The storage miners also produce PoSs, which are utilized to validate the weights of miners in SW-BFT consensus. The retrieval miner collects chunks from storage miners to recover files.
We built the encoding miner module, retrieval miner module, HF module and PoS module from scratch, marked green in Fig.~\ref{fig:diagram}. The storage miner module is inherited from Filecoin Lotus, the most popular implementation of Filecoin, which is the biggest DSN\footnote{https://spec.filecoin.io/implementations}. The storage miner module is colored blue in Fig.~\ref{fig:diagram}.
We use Tendermint Core, an open-source BFT consensus engine, as the consensus engine of BFT-DSN, colored orange in Fig.~\ref{fig:diagram}. BFT-DSN empploys  the storage-weighted BFT consensus, SW-BFT, which adjusts the voting power of each miner based on the amount of storage resources it contributes to the network. To achieve this, we examined the source code of Lotus and found the function used to update the storage power table, $\mathsf{UpdatePower}$. We added some code at the entry of $\mathsf{UpdatePower}$ to call $\mathsf{ValidatorUpdate}$, the function used by Tendermint Core, to adjust each miner's voting power. 
Additionally, the RS coding library is provided by Klaus Post\footnote{https://github.com/klauspost/reedsolomon}, implemented over a Galois Field with order 256. The weighted threshold signature library is given in \cite{weightedTS} \footnote{https://github.com/sourav1547/wts}. We also implemented the homomorphic fingerprinting algorithm \cite{verifiableEC} based on the RS coding library.

\begin{figure}[htbp]
	\centering\includegraphics[width=0.48\textwidth]{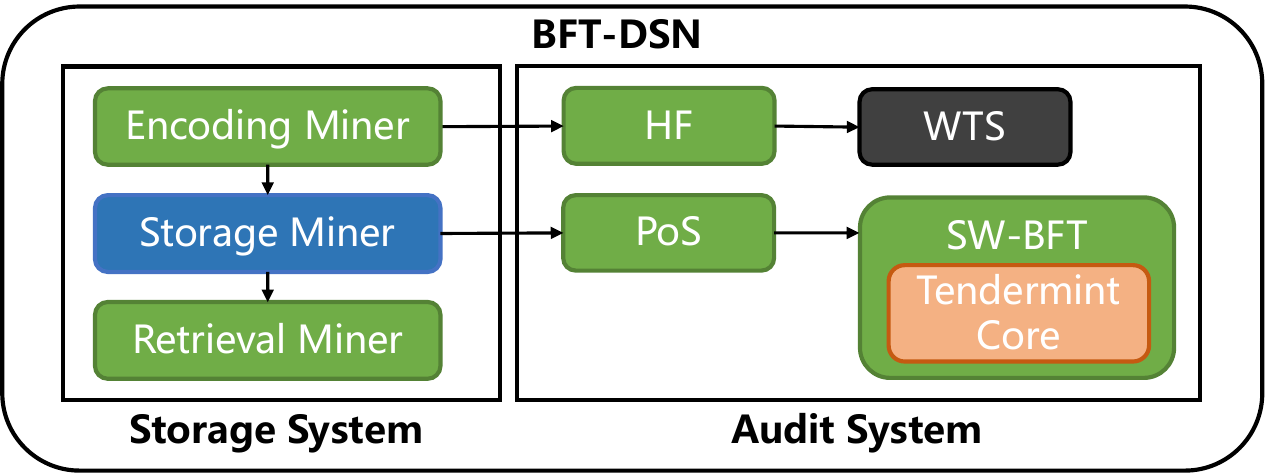}
	\caption{The Block diagram of BFT-DSN}
	\label{fig:diagram}
\end{figure}

\subsection{Experiment Setup}

In our evaluation, we tested BFT-DSN using networks of varying sizes, ranging from 10 sectors to 100 sectors for 2 to 100 nodes. To construct a BFT-DSN network with $n$ sectors, we had $n$ nodes pledge sectors following a Poisson process until a total of $n$ sectors were pledged. Nodes that did not pledge any sector were subsequently removed.
Each node was equipped with four CPU cores operating at a speed of 2.10 GHz, 8 GB of RAM, and 1 TB of storage capacity. The network had a data transfer rate of 1 GBps. As mentioned in Section~\ref{sec:design}, sectors in BFT-DSN have equal sizes. In our evaluation, this size was set to 32 GB to align with Filecoin's common deployment.
To assess the system, we generated random files of different sizes, ranging from 10 MB to 200 MB, using the /dev/urandom file in Linux.

To demonstrate that BFT-DSN achieves optimal Byzantine resilience without sacrificing performance, we focus on four crucial metrics: Byzantine resiliency, storage costs, put latency, and get latency. Evaluation results were averaged over 100 trials. For evaluation purposes, we compared BFT-DSN with three widely-used DSNs \cite{dsnsurvey}, namely Filecoin, Sia, and Storj mentioned in Section~\ref{sec:related}, which served as the baselines. Specifically, since Sia and Storj use EC for storage redundancy, we set their EC scheme to be the same as BFT-DSN, which is $(n-f, f)$-RS. It is important to note that Sia and Storj do not allow storing more than one chunk of a file in a single miner. Therefore, we distribute each of the $n$ chunks encoded from a file to $n$ miners, with each chunk assigned to a different miner. 
Besides, Filecoin achieves duplicated storage by generating copies. Thus, to tolerate $f$ Byzantine faults from $n$ sectors, we let Filecoin store $f+1$ copies of each file.

\subsection{Evaluation Results}

\textbf{Byzantine Resiliency.}
First, we evaluate the Byzantine resilience of our BFT-DSN by measuring the success rate of the retrieval process.  We simulate Byzantine adversaries that actively attack the storage service and consensus mechanism. As shown in Figure~\ref{fig:availability}, when less than $\frac{1}{3}$ of the storage resources are controlled by Byzantine adversaries, the retrieval process of a file in BFT-DSN always succeeds. This is because both EC-based storage and SW-BFT consensus of BFT-DSN tolerate up to $\frac{1}{3}$ proportion of Byzantine sectors, and verifiability throughout a file's lifecycle is guaranteed by HF and WTS. Filecoin, which uses storage redundancy based on replication, achieves a near 100\% success rate when less than 20\% of storage resources being controlled by Byzantine adversaries. Note that the success rate is not 100\%  because there exist cases when all replicas of a file happen to be stored by Byzantine miners (the file would become unavailable). When 20\% or more storage resources are controlled by Byzantine adversaries, consensus fails \cite{ECConsensusAnal}, and the Byzantine adversaries are able to block the retrieval requests sent by the client and completely obstruct retrieval. In Sia and Storj, redundant storage is based on erasure coding, but the integrity of the collected chunks for decoding is not verified. This allows Byzantine miners to provide modified chunks to prevent the decoder from decoding the requested file. Only when all the first $n-f$ arriving chunks are sent by honest miners (the probability of which decreases rapidly as the proportion of Byzantine storage resources increases) would the correct file be output.

\begin{figure}[htbp]
    \centering
    \includegraphics[width=0.48\textwidth]{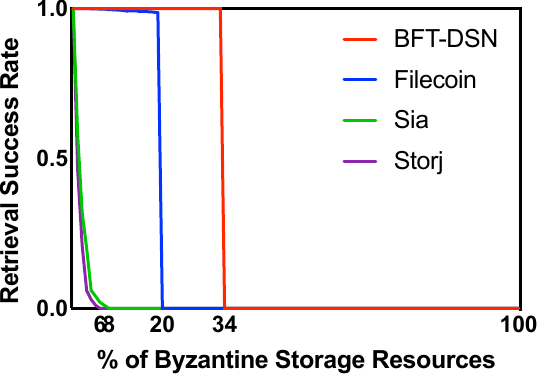}
    \caption{File availability with different proportion of storage resources controlled by Byzantine adversaries}
    \label{fig:availability}
\end{figure}

\textbf{Storage Cost.}
Next, we evaluate the storage cost of BFT-DSN. As shown in Figure~\ref{fig:storage}, when $n=40$, the storage costs grow linearly with the increase of file size in all four DSNs\footnote{We varied $n$ in our experiments and obtained results with the same trend. Therefore in this section, we only report the results when $n=40$.}. However, Filecoin's storage costs grow much faster than the other three DSNs that use EC. For a 200 MB file, Filecoin requires more than 3.5 GB of storage, while the other DSNs, including BFT-DSN, require less than 0.5 GB. To determine the average storage costs for files, we apply linear regression to storage costs for files ranging from 10 MB to 200 MB. When $n=40$, the $f$-replica storage costs an average of 20.63 MB of space per MB of file, while BFT-DSN, Sia, and Storj cost 1.49 MB, 1.55 MB, and 1.48 MB of space respectively. This is because EC provides storage redundancy with a significantly lower cost compared to directly making copies. The average storage costs per MB of Sia and Storj are close to BFT-DSN, which are 1.55 MB and 1.48 MB, respectively. The storage cost of Filecoin step-increases around a linear function of file size due to Filecoin's file padding mechanism. Each file stored in Filecoin is padded to the nearest power of 2 in bytes before undergoing further processing.

\begin{figure}[htbp]
	\centering
		\includegraphics[width=0.48\textwidth]{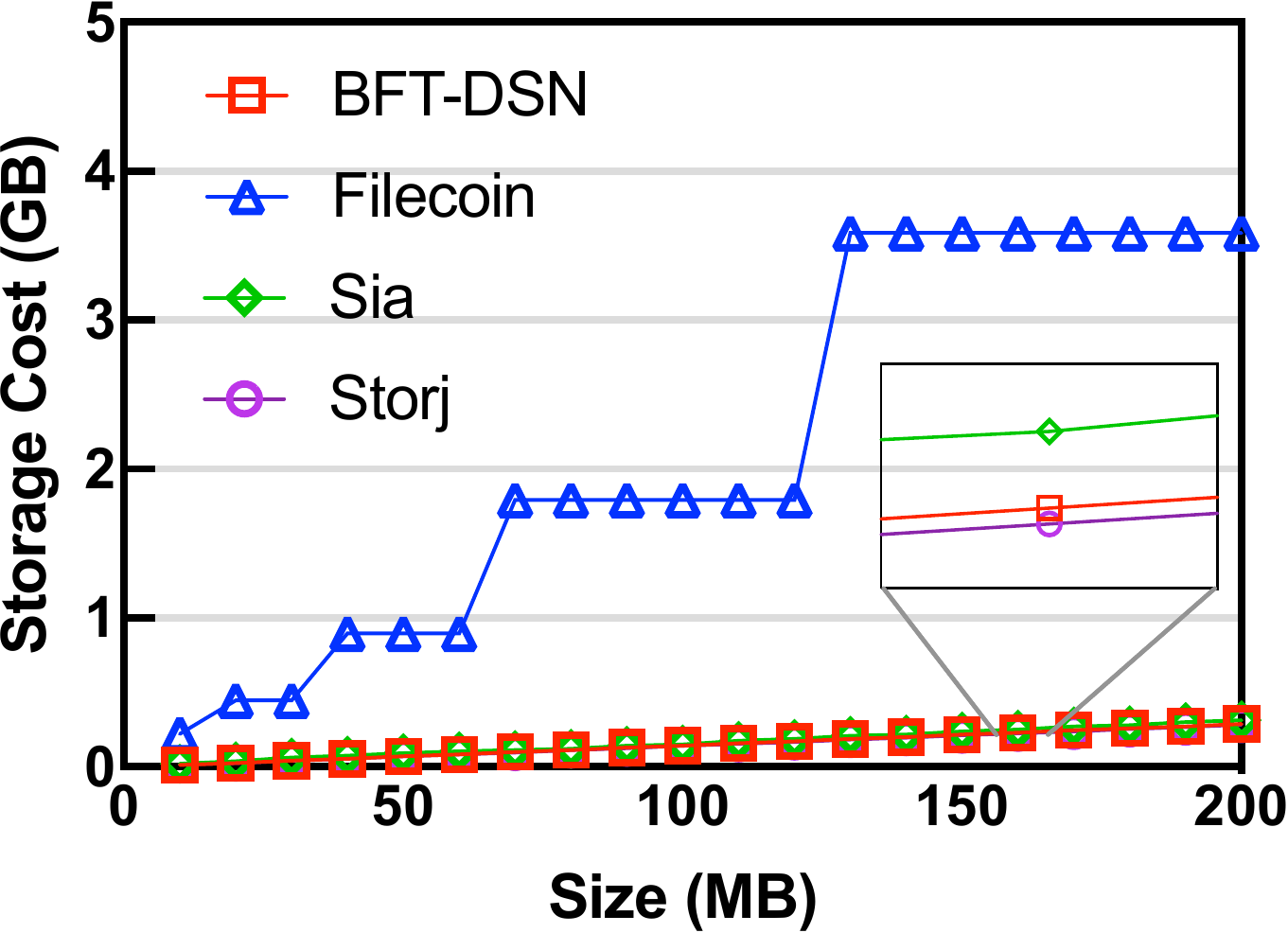}
	\caption{Storage costs, $n$=40}
	\label{fig:storage}
\end{figure}

% \begin{figure}[htbp]
% 	\centering
%         \begin{subfigure}{6cm}
% 			\includegraphics[width=\textwidth]{images/legend.pdf}
% 	\end{subfigure}
% 	\qquad
% 	\\
% 	\begin{subfigure}{3.97cm}
% 		\includegraphics[width=\textwidth]{images/storage_overhead_size.pdf}
% 		\caption{Storage costs of different file size, $n$=40}
% 		\label{fig:storage_size}
% 	\end{subfigure}
% 	\qquad
% 	\begin{subfigure}{3.97cm}
% 		\includegraphics[width=\textwidth]{BFT-DSN_tex/images/storage_overhead_per_MB.pdf}
% 		\caption{Storage costs per MB with different $n$}
% 		\label{fig:storage_n}
% 	\end{subfigure}
% 	\caption{Storage Costs}
% 	\label{fig:storage}
% \end{figure}

\begin{figure*}[!hbtp]
	\centering
     \begin{subfigure}{6cm}
			\includegraphics[width=\textwidth]{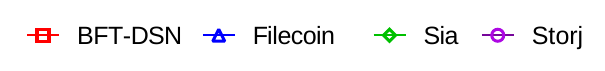}
	\end{subfigure}
	\qquad
	\\
	\begin{subfigure}{3.98cm}
		\includegraphics[width=\textwidth]{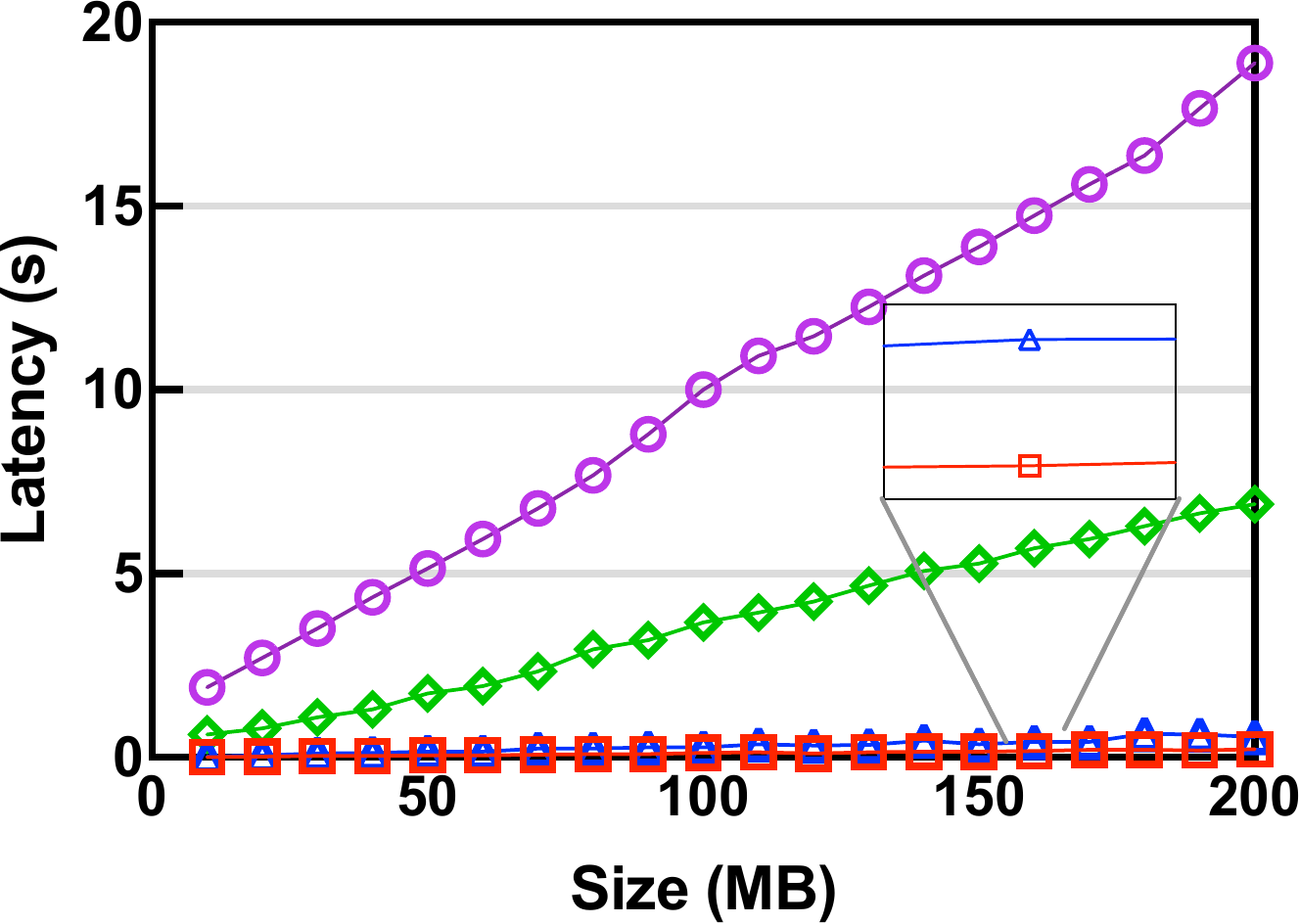}
		\caption{Put latency with different file size, $n=10$}
		\label{fig:put_latency_n10}
	\end{subfigure}
	\qquad
	\begin{subfigure}{3.98cm}
		\includegraphics[width=\textwidth]{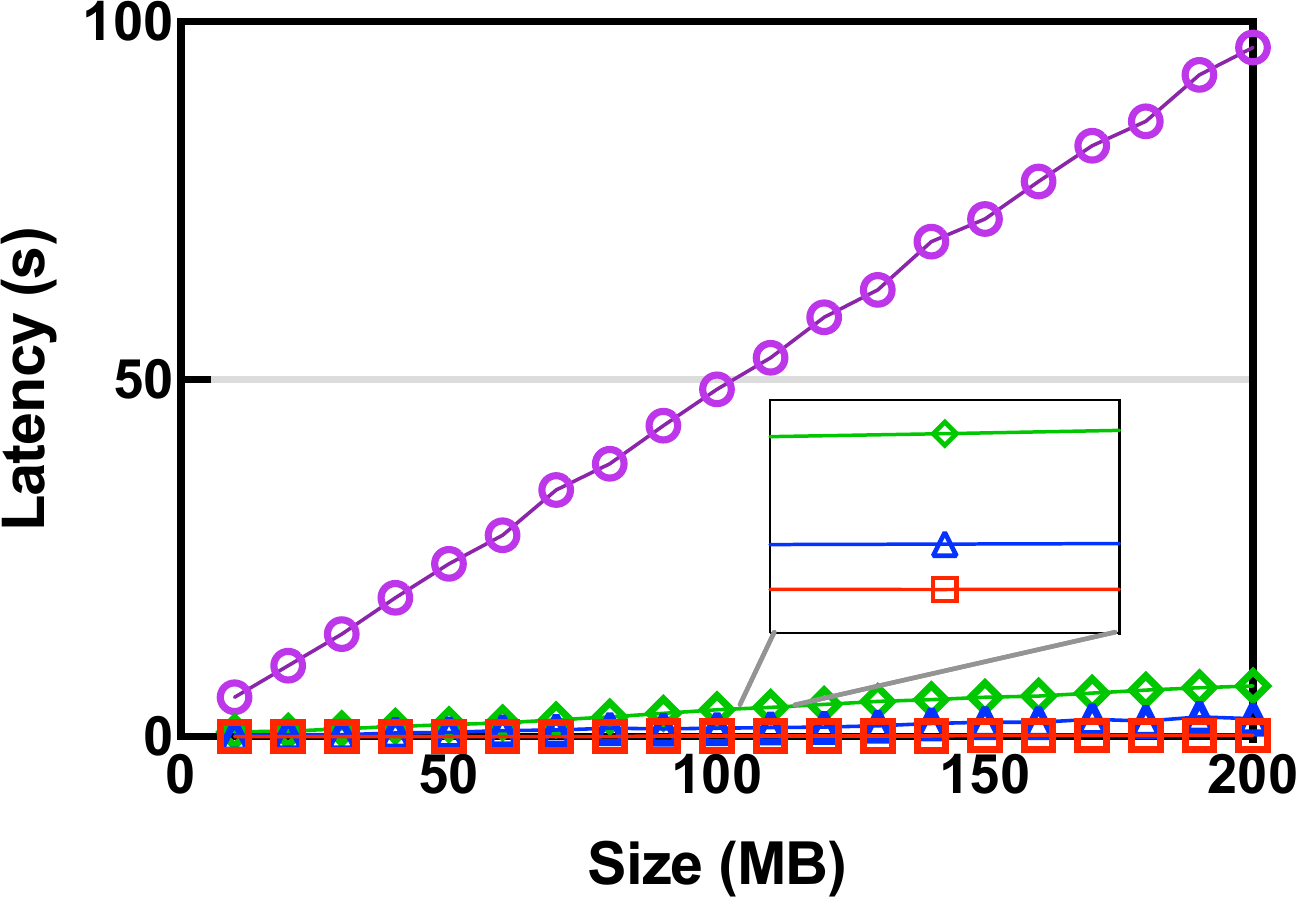}
		\caption{Put latency with different file size, $n=40$}
		\label{fig:put_latency_n40}
	\end{subfigure}
        \qquad
	\begin{subfigure}{3.98cm}
		\includegraphics[width=\textwidth]{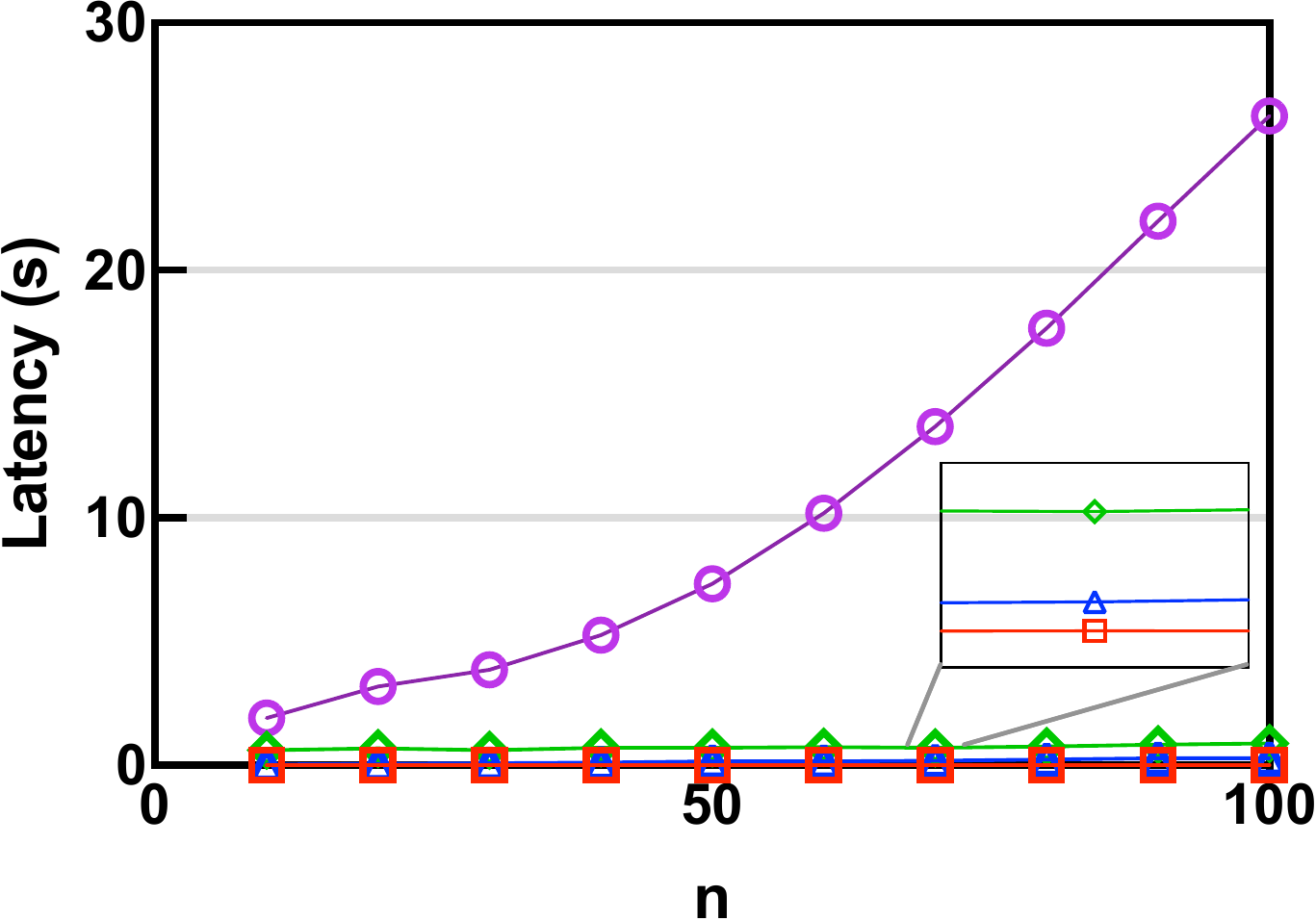}
		\caption{Put latency with different $n$, $|\mathcal{F}|=$10~MB}
		\label{fig:put_latency_size10}
	\end{subfigure}
        \qquad
	\begin{subfigure}{3.98cm}
		\includegraphics[width=\textwidth]{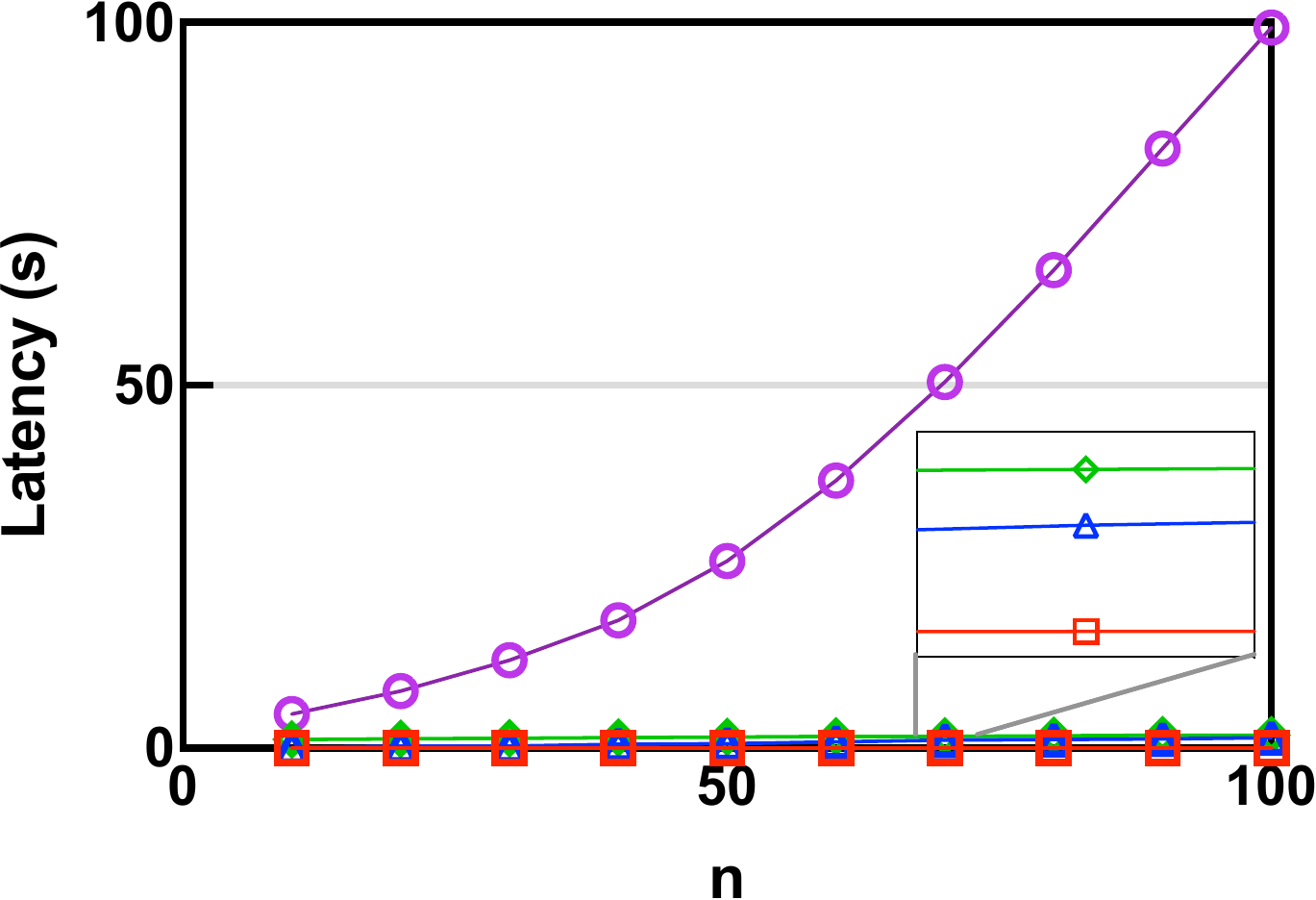}
		\caption{Put latency with different $n$, $|\mathcal{F}|=$40~MB}
		\label{fig:put_latency_size40}
	\end{subfigure}
	\caption{Put latency}
	\label{fig:put_latency}
\end{figure*}

\begin{figure*}[!htbp]
	\centering
        \begin{subfigure}{6cm}
			\includegraphics[width=\textwidth]{images/legend.pdf}
	\end{subfigure}
	\qquad
	\\
	\begin{subfigure}{3.98cm}
		\includegraphics[width=\textwidth]{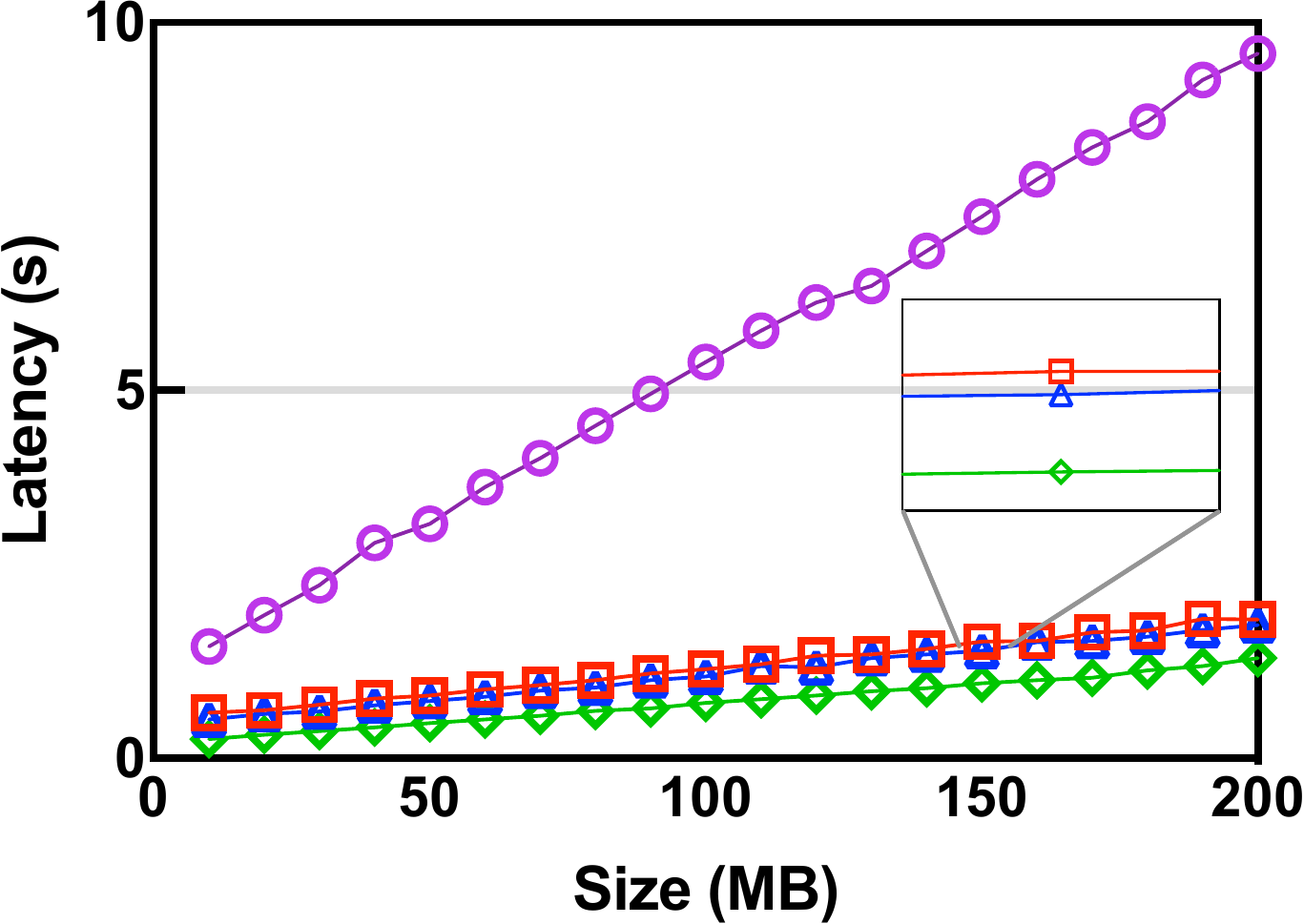}
		\caption{Get latency with different file size, $n=10$}
		\label{fig:get_latency_n10}
	\end{subfigure}
	\qquad
	\begin{subfigure}{3.98cm}
		\includegraphics[width=\textwidth]{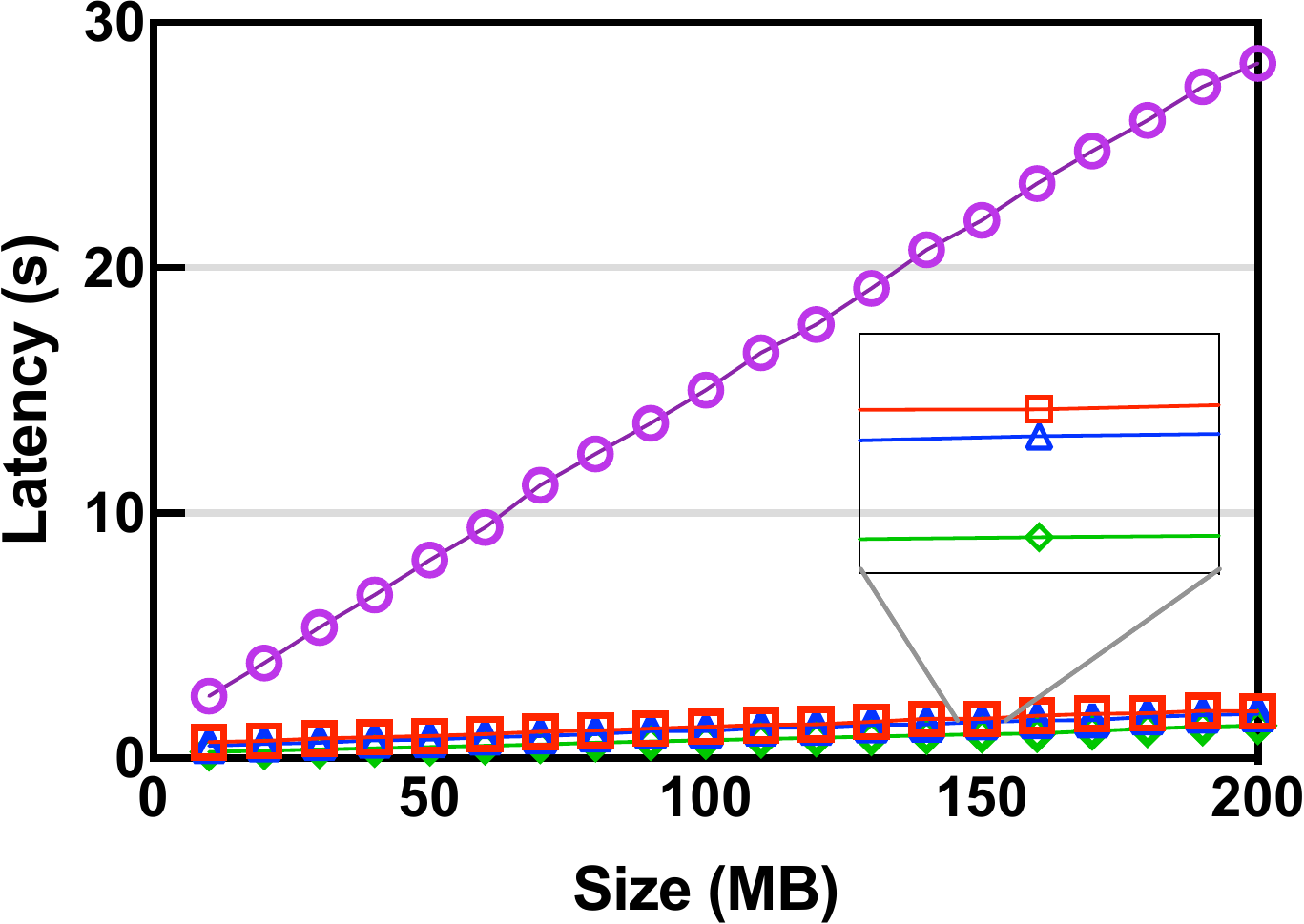}
		\caption{Get latency with different file size, $n=40$}
		\label{fig:get_latency_n40}
	\end{subfigure}
        \qquad
	\begin{subfigure}{3.98cm}
		\includegraphics[width=\textwidth]{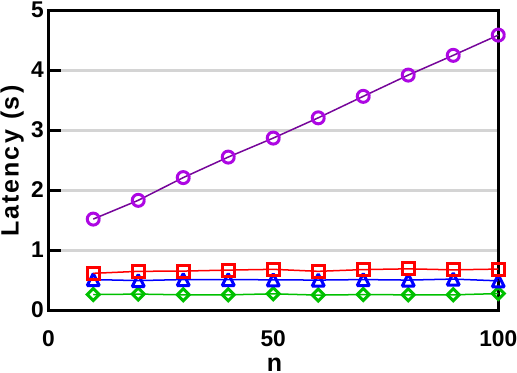}
		\caption{Get latency with different $n$, $|\mathcal{F}|=$10~MB}
		\label{fig:get_latency_size10}
	\end{subfigure}
        \qquad
	\begin{subfigure}{3.98cm}
		\includegraphics[width=\textwidth]{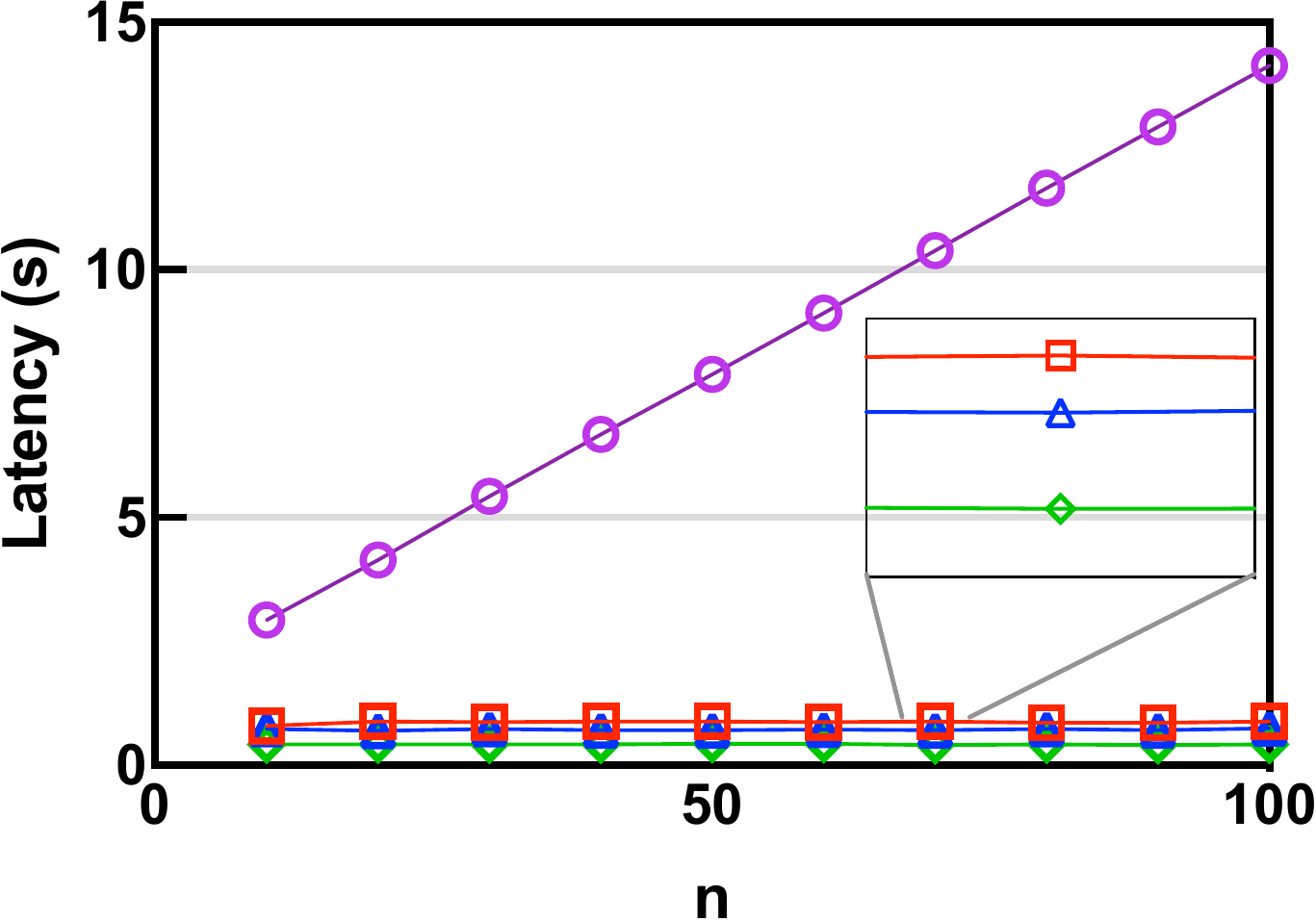}
		\caption{Get latency with different $n$, $|\mathcal{F}|=$40~MB}
		\label{fig:get_latency_size40}
	\end{subfigure}
	\caption{Get latency}
	\label{fig:get_latency}
\end{figure*}

\textbf{Put Latency.} Then, we evaluate the put latency of BFT-DSN. In our experiments, the measurement started from when sending the storage request and ended when all $n$ chunks were stored for BFT-DSN, Sia, and Storj. For Filecoin, the measurement started from when sending the storage request and ended when $f+1$ copies of the file were stored. The results are displayed in Fig.~\ref{fig:put_latency}.
First, we investigate how file size affects the put latency of BFT-DSN and the baselines in two different settings: one with $n$ set to 10, and another with $n$ set to 40. As shown in Fig.~\ref{fig:put_latency_n10} and Fig.~\ref{fig:put_latency_n40}, the put latencies of all four DSNs increase linearly with file size. Due to extra padding and encryption processes, the put latencies of Storj and Sia grow faster than the other two. Additionally, the put latency of Storj grows even faster than that of Sia due to the heavy load of scheduling. BFT-DSN's put latency leads by a slim margin compared to Filecoin. %This is because both the RS-encoding of a file and the distribution of the chunks have a linear complexity that depends on the file size.
Then, we examine how the value of $n$ affects the put latencies. The results are shown in Fig.~\ref{fig:put_latency_size10} and Fig.~\ref{fig:put_latency_size40} for file sizes of 10 MB and 40 MB, respectively. For BFT-DSN, Filecoin, and Sia, the increase in put latency due to the increase in $n$ is negligible. However, for Storj, the put latency grows in a super-linear manner. This is because Storj uses a satellite node to schedule the upload of a file. During the upload process, the satellite is responsible for obtaining a list of $n$ storage miners. To obtain this list, the satellite checks the authorization of each of the $n$ miners by reading information from a database. The time complexity of reading information from a database of $n$ miners is $O(\log n)$. Therefore, the time complexity of the Storj satellite obtaining the list of $n$ storage miners is at least $O(n\log n)$, which is super-linear.

\textbf{Get Latency.}
In our experiments, we measured the latency of the get operation from the moment the retrieval request was sent until the requested file was recovered. The evaluation results are shown in Fig.~\ref{fig:get_latency}.
As seen in Fig.~\ref{fig:get_latency_n10} and Fig.~\ref{fig:get_latency_n40}, both the get latency of BFT-DSN and that of $f$-replica Filecoin increase linearly with the file size $|\mathcal{F}|$. This is because the transmission time of both storage methods grows linearly w.r.t. the file size. BFT-DSN's get latency is similar to that of Filecoin and Sia, while Storj's get latency increases much faster due to chunk padding.
To examine how the network size affects the get latency of BFT-DSN and the baselines, we evaluate their put latency with fixed file sizes and vary $n$, as shown in Fig.~\ref{fig:get_latency_size10} and Fig.~\ref{fig:get_latency_size40}. With fixed file sizes, the get latencies of BFT-DSN, Filecoin, and Sia remain constant as $n$ increases, while Storj's get latency grows linearly with $n$ due to extra scheduling process.

\section{Conclusions}\label{sec:conclusion}
In this paper, we present the design and implementation of BFT-DSN, a DSN that achieves optimal Byzantine resilience. BFT-DSN combines publicly verifiable EC-based storage with storage-weighted BFT consensus, unifying their security upper bound to $\lfloor\frac{n-1}{3}\rfloor$, which is optimal. Moreover, BFT-DSN enhances the decentralized verification of erasure coding through the adoption of homomorphic fingerprints and weighted threshold signatures. We have implemented a working instance of BFT-DSN and evaluated its performance. The results show that BFT-DSN delivers performance comparable to the state-of-the-art industrial DSNs in terms of storage cost and latency while offering better Byzantine fault tolerance. In conclusion, BFT-DSN achieves optimal Byzantine resilience while maintaining performance comparable to existing industrial DSNs.
One limitation of this work is that we did not apply erasure coding to all data in the DSN, such as blockchain and state data. In future research, we will explore ways to optimize storage costs in DSNs, allowing DSN technologies to be applied in a wider range of scenarios.
	%\newpage

\section{Acknowledgement}
This study was partially supported by the National Key R\&D Program of China (No.2022YFB4501000), the National Natural Science Foundation of China (No.62232010, 62302266, U23A20302), Shandong Science Fund for Excellent Young Scholars (No.2023HWYQ-008), Shandong Science Fund for Key Fundamental Research Project (ZR2022ZD02), and the Fundamental Research Funds for the Central Universities.
	
	\ifCLASSOPTIONcaptionsoff
	\newpage
	\fi
	
	\bibliographystyle{IEEEtran}
	\bibliography{bftdsn}
	
\end{document}